\documentclass[11pt]{article}
\usepackage[T1]{fontenc}
\usepackage[margin=1in]{geometry} 
\usepackage{libertineRoman}
\usepackage{biolinum}
\usepackage{libertineMono} 
\usepackage[utf8]{inputenc} 
\usepackage[normalem]{ulem} 
\usepackage{changepage} 

\usepackage{tikz}
\tikzset{ 
  facehighlight/.style={fill=gray,opacity=0.25},
    darkedge/.style={draw=black!50, line width=0.5pt}
}
\usepackage{amsmath,amssymb}
\usetikzlibrary{positioning,arrows.meta}
\usetikzlibrary{arrows.meta}
\usepackage{graphicx} 

\usepackage{cite}
\usepackage{pgfplots}
\pgfplotsset{compat=1.18}

\usepackage{algorithm}
\usepackage{algpseudocode}
\usepackage{amsfonts}
\usepackage{amsmath}
\usepackage{amssymb}
\usepackage{amsthm}
\usepackage{authblk}
\usepackage{bbm}
\usepackage{enumerate}
\usepackage{graphicx}
\usepackage{ifthen}
\usepackage{latexsym}
\usepackage{mathtools}
\usepackage{multicol}
\usepackage{mathpazo}

\usepackage{thmtools} 
\usepackage{tabto}
\usepackage[colorlinks=true,linkcolor=blue,anchorcolor=blue,citecolor=red,urlcolor=magenta]{hyperref}
\usepackage[capitalize]{cleveref} 

\crefname{appendix}{Appendix}{Appendices}
\Crefname{appendix}{Appendix}{Appendices}

\usepackage{color, xcolor}
\usepackage{stmaryrd} 
\usepackage{makeidx} 

\crefname{equation}{}{}

\newtheorem{theorem}{Theorem}

\newtheorem{definition}[theorem]{Definition}
\newtheorem{lemma}[theorem]{Lemma} 
\newtheorem{corollary}[theorem]{Corollary}
\newtheorem{remark}[theorem]{Remark}
\newtheorem{proposition}[theorem]{Proposition}

\newcommand{\tinyspace}{\mspace{1mu}}

\newcommand{\tr}{\operatorname{Tr}}
\newcommand{\rank}{\operatorname{rank}}

\newcommand{\im}{\operatorname{im}}  

\newcommand{\ket}[1]{| #1 \rangle}
\newcommand{\braket}[2]{\langle #1 | #2 \rangle}
\newcommand{\ketbra}[2]{| #1 \rangle\langle #2 |}

\newcommand{\norm}[1]{\left\lVert\tinyspace#1\tinyspace\right\rVert}
\newcommand{\abs}[1]{\left\lvert\tinyspace #1 \tinyspace\right\rvert}
\newcommand{\floor}[1]{\left\lfloor #1 \right\rfloor}
\newcommand{\ceil}[1]{\left\lceil #1 \right\rceil}

\renewcommand{\t}{{\scriptscriptstyle\mathsf{T}}}

\def\I{\mathbb{1}} 
 
\def\C{\mathbb{C}}

\def\R{\mathbb{R}}

\def\d{\mathbf{d}}
\def\p{\mathbf{p}}
\def\vv{\mathbf{v}}
\def\x{\mathbf{x}}
\def\y{\mathbf{y}}

\def\spn{\operatorname{span}}

\newcommand{\e}{\mathbf{e}}
\newcommand{\w}{\mathbf{w}}

\newcommand{\one}{\mathbf 1}
\newcommand{\supp}{\mathrm{supp}}
\newcommand{\vareps}{\varepsilon}

\title{
Sharp bounds for perfect quantum state classification beyond antidistinguishability  
}

\newcommand{\authemail}[1]{{\fontsize{7.5}{8.5}\selectfont\ttfamily #1}}
\author{Nathaniel Johnston\thanks{\scriptsize Department of Mathematics \& Computer
Science, Mount Allison University. \authemail{njohnston@mta.ca}}, \; 
Benjamin Lovitz\thanks{\scriptsize Department of Computer Science and Software
Engineering, Concordia University. \authemail{benjamin.lovitz@concordia.ca}}, \; 
Vincent Russo\thanks{\scriptsize Unitary Foundation and IonQ. \authemail{vincentrusso1@gmail.com}}, \; 
Jamie Sikora\thanks{\scriptsize Department of Computer Science, Virginia Tech.
\authemail{sikora@vt.edu}} 
} 

\date{September 16, 2026} 

\begin{document}

\maketitle 

\begin{abstract}
    A multiset of pure quantum states is said to be \emph{$k$-learnable} if there is a measurement strategy that always narrows an unknown sample drawn from the list down to one of at most $k$ candidates. The parameter $k$ interpolates between distinguishability and antidistinguishability, and provides a unified framework for partial state identification.
    We prove two universal, and optimal, Gram-matrix criteria for $k$-learnability: 
    a Frobenius-norm sufficient condition and an entrywise-$\ell_1$ necessary condition. 
    We apply them to derive explicit learnability and copy-complexity guarantees for several well-known sets of states including SIC-POVMs, mutually unbiased bases, and stabilizer states. 
    We further apply our results to zero-error mutation detection problems such as anomaly detection and changepoint detection.  
\end{abstract}

\section{Introduction} \label{sec:intro}

A central problem in quantum information is to extract classical information from quantum states. 
The most basic version asks: given a single copy of an unknown pure state drawn from a known finite multiset $S = \{\ket{\psi_1}, \ket{\psi_2}, \ldots, \ket{\psi_n}\}$, can a measurement be designed that reliably reveals which state was drawn? 
When such a measurement always succeeds in this task, we say that $S$ is \emph{distinguishable}; otherwise some pairs of states inevitably get confused for one another. 
Between the two extremes of perfect distinguishability and complete indistinguishability lies a rich landscape of partial-information tasks: for instance, one may ask not for the identity of the state but only to \emph{rule out} some candidates, leading to the notion of \emph{antidistinguishability}~\cite{caves2002conditions, pusey2012reality, bandyopadhyay2014conclusive, leifer2014ontology, heinosaari2018antidistinguishability, havlicek2020simple, russo2023inner, mishra2024optimal, johnston2023antidistinguishability, johnston2025complexity}. 
The concept was originally introduced by Caves, Fuchs, and Schack under the name \emph{post-Peierls incompatibility}~\cite{caves2002conditions}  and an equivalent formulation in terms of \emph{perfect quantum state exclusion} was given by Bandyopadhyay, Jain, Oppenheim, and Perry~\cite{bandyopadhyay2014conclusive}. 
It has also played a central role in the celebrated PBR theorem~\cite{pusey2012reality}. 

In this paper, we study a common refinement of these notions called \emph{$k$-learnability},  which we introduced in~\cite{johnston2025complexity}. 
Informally, $S$ is $k$-learnable if a measurement can be designed whose outcome always narrows the identity of the unknown state down to one of at most $k$ candidates. 
The case $k=1$ is equivalent to the task of distinguishability and the case $k=n-1$ is equivalent to the task of antidistinguishability. 
This task for intermediate values of $k$ is a discrete interpolation between the two. 
Perhaps surprisingly, the states need not be pairwise orthogonal to accomplish this task perfectly: there exist non-orthogonal pure states that are $k$-learnable for $k \geq 2$. 
Understanding the conditions under which a multiset of pure states is $k$-learnable is the focus of this work. 

\begin{definition}
Let $n \geq 2$ be a positive integer and let $k \in \{ 1, \ldots, n \}$. 
A multiset of pure quantum states $\{ \ket{\psi_1}, \ldots, \ket{\psi_n} \} \subset \mathbb{C}^d$  
is \emph{$k$-learnable} if there exists a POVM
\begin{equation}
\{ M_I : I \subseteq \{1, 2, \ldots, n\},\, |I| = k \},
\end{equation}
such that
\begin{equation} 
\langle \psi_i | M_I | \psi_i \rangle = 0
\qquad \text{whenever } i \notin I.
\end{equation} 
In other words, if $\ket{\psi_i}$ is measured with such a POVM, only the subsets of size $k$ containing $i$ could occur. 
\end{definition} 

\subsection{Main results} 

We now discuss our main results and their extensions, corollaries, and applications.   

\paragraph{Optimal bounds.} 
We study this new classification task by asking two complementary questions. 

\begin{quote}
\textit{What conditions on the states guarantee $k$-learnability? 
What conditions on the states forbid $k$-learnability?} 
\end{quote} 

To this end, we provide two complementary bounds (proofs in the appendix). 

\begin{theorem}\label{thm:main} 
    Let $n \geq 2$ be an integer, $k \in \{1,2,\ldots, n \}$, and let $S = \{\ket{\psi_1}, \ket{\psi_2}, \ldots, \ket{\psi_n} \}$. 
    If 
    \begin{equation}\label{eq:thm_main_frob}
        \sum_{i, j = 1}^{n} \abs{\braket{\psi_i}{\psi_j}}^2 \leq \frac{n^2}{n-k+1}, 
    \end{equation} 
    then $S$ is $k$-learnable. 
    In particular, if
    \begin{equation} 
        \label{eq:thm_main_coh}
        \abs{\braket{\psi_i}{\psi_j}} \leq \sqrt{\frac{k-1}{(n-1)(n-k+1)}} \quad \text{for all} \quad i \neq j 
    \end{equation}
    then $S$ is $k$-learnable.
\end{theorem}  

This theorem shows that if the inner products are not collectively \emph{large}, then we can guarantee $k$-learnability. 
Of significance, we do not require pairwise orthogonality when $k \geq 2$. The left-hand side of~\eqref{eq:thm_main_frob} is called the \emph{frame potential} of $S$~\cite{benedetto2003finite}. The case $S \subset \C^{n-k+1}$ is particularly interesting. In this case, it is known that the right-hand side of~\eqref{eq:thm_main_frob} is the smallest possible frame potential of any set of $n$ states~\cite{benedetto2003finite} (see also the \emph{Welch bound}~\cite{welch1974lower}).  In particular, our theorem shows that any set of $n$ pure states in  $\C^{n-k+1}$ attaining the minimum possible frame potential is $k$-learnable. 

On the other hand, if the states are \emph{too close together}, they cannot be $k$-learnable. 

\begin{theorem}\label{thm:m_excl_threshold}
    Let $n \geq 2$ be an integer, $k \in \{1, 2, \ldots, n\}$, and let $S = \{\ket{\psi_1}, \ket{\psi_2}, \ldots, \ket{\psi_n} \}$. 
    If
    \begin{equation}\label{eq:inner-product-non-mexcl-threshold}
        \sum_{i,j = 1}^{n} \abs{\braket{\psi_i}{\psi_j}} > nk,
    \end{equation}
    then $S$ is not $k$-learnable. 
    In particular, if
    \begin{equation}
        \abs{\braket{\psi_i}{\psi_j}} > \frac{k-1}{n-1} \quad \text{for all} \quad i \neq j
    \end{equation}
    then $S$ is not $k$-learnable. 
\end{theorem} 

Together, the above two bounds give explicit criteria for when pure states must be or cannot be $k$-learnable.  

\textit{A remark about multisets.}  
While we sometimes describe the task of state discrimination as ``determining which state was sent'', a more accurate description is to ``determine the index of the state that was sent.'' 
This is because we allow multiple copies of a single state, or those differing by only a global phase, to be within the multiset. 
For example, for the multiset $\{ \ket{\psi_1} = \ket{0}, \ket{\psi_2} = \ket{0} \}$, determining which state was sent is easy, but determining the index of the state that was sent is hard.  

\textit{A remark about the dimensions of the pure states.} 
We note that sometimes we specify the dimension of the pure states that we wish to $k$-learn and sometimes it is omitted, like in Theorem~\ref{thm:m_excl_threshold}. 
Many of our results rely only on the Gram matrix $G$ of the multiset of states which is defined as 
\begin{equation} 
G_{i,j} = \braket{\psi_i}{\psi_j}. 
\end{equation}  
When a multiset of states is described only by its Gram matrix, this does not fully specify the states (although they are guaranteed to exist).  
We note that the rank of the Gram matrix is equal to the minimum dimension of any collection of vectors that have that Gram matrix, so there is an implicit bound on the dimension required to realize such states. 

\paragraph{Simplex states and $2$-learnability.} 

Since perfect distinguishability is a strict condition (pairwise orthogonality), the next best scenario is being $2$-learnable. 
In other words, one can correctly identify the state with $2$ guesses, without error. 
It turns out that the case of $k=2$ is interesting in light of the two previously mentioned results. 
From Theorem~\ref{thm:main}, we have that $S = \{\ket{\psi_1}, \ket{\psi_2}, \ldots, \ket{\psi_n} \}$ is $2$-learnable when 
\begin{equation} 
| \braket{\psi_i}{\psi_j} |\leq \frac{1}{n-1}
\end{equation} 
for all $i \neq j$ 
and from Theorem~\ref{thm:m_excl_threshold}, we have that $S$ cannot be $2$-learnable when 
\begin{equation}
| \braket{\psi_i}{\psi_j}| > \frac{1}{n-1}
\end{equation} 
for all $i \neq j$. 
Thus, when $|\braket{\psi_i}{\psi_j}| = \frac{1}{n-1}$ for $i \neq j$ we have a multiset of states that is $2$-learnable, but just barely. 

Consider the multiset of states $S = \{\ket{\psi_1}, \ket{\psi_2}, \ldots, \ket{\psi_n} \} \subseteq \C^{n-1}$ which satisfies 
\begin{equation} 
\braket{\psi_i}{\psi_j} = - \frac{1}{n-1} \quad \text{for all} \quad i \neq j 
\end{equation}  
called the \emph{simplex states}.  
Note that we explicitly mentioned the dimension as $n-1$ since the Gram matrix has rank $n-1$. 
Thus, we have $n$ states in $n-1$ dimensions that are not even linearly independent but still \emph{spread out enough} to be $2$-learnable.  

The simplex states saturate all the inequalities in both of Theorems~\ref{thm:main} and \ref{thm:m_excl_threshold} with equality. 
If \emph{any} inner product were to increase in magnitude (by possibly increasing the dimension), then they would cease to be $2$-learnable, but if \emph{any} inner product were to decrease in magnitude then $2$-learnability is still guaranteed. 

Simplex states generalize the \emph{trine states} when $n=3$ and the \emph{tetrahedral states} when $n=4$, both depicted below.  

\begin{figure}[!htpb]
\centering
\begin{tikzpicture}[
    vector/.style={line width=0.9pt},
    frame/.style={gray!45, line width=0.4pt},
    edge/.style={gray, line width=0.5pt, opacity=0.45},
    face/.style={gray, opacity=0.12},
    hidden/.style={gray!40, line width=0.5pt, dash pattern=on 2pt off 1.6pt},
    angarc/.style={gray!65, line width=0.5pt},
    dot/.style={circle, fill=black, inner sep=0pt, minimum size=3.4pt},
    slabel/.style={font=\small, inner sep=2pt},
    axlabel/.style={font=\footnotesize, color=gray!70, inner sep=2pt},
    arclabel/.style={font=\footnotesize, color=gray!70, inner sep=1.5pt, fill=white},
    panel/.style={font=\small\bfseries, anchor=west}]

\begin{scope}
    \draw[frame] (0,0) circle (2.200);
    \draw[frame] (-2.550,0) -- (3.050,0);
    \draw[frame] (0,-2.550) -- (0,2.800); 
    \fill[face] (2.200,0.000) -- (-1.100,1.905) -- (-1.100,-1.905) -- cycle;
    \draw[edge] (2.200,0.000) -- (-1.100,1.905) -- (-1.100,-1.905) -- cycle;
    \node[dot] at (2.200,0.000) {};
    \node[slabel, anchor=south west] at (2.300,0.000) {$\ket{\psi_1}$};
    \node[dot] at (-1.100,1.905) {};
    \node[slabel, anchor=south east] at (-1.150,1.992) {$\ket{\psi_2}$};
    \node[dot] at (-1.100,-1.905) {};
    \node[slabel, anchor=north east] at (-1.150,-1.992) {$\ket{\psi_3}$}; 
\end{scope}

\begin{scope}[xshift=8.3cm,yshift=-0.35cm,scale=1.1] 
    \draw[hidden] (-1.333,-0.085) -- (2.043,-0.545);
    \node[dot] at (-0.000,2.040) {};
    \node[dot] at (-1.333,-0.085) {};
    \node[dot] at (2.043,-0.545) {}; 
    \node[dot] at (-0.709,-1.410) {};
    \node[slabel, anchor=south] at (-0.000,2.140) {$\ket{\psi_1}$};
    \node[slabel, anchor=east] at (-1.433,-0.091) {$\ket{\psi_2}$};
    \node[slabel, anchor=west] at (2.139,-0.571) {$\ket{\psi_3}$};
    \node[slabel, anchor=north] at (-0.754,-1.499) {$\ket{\psi_4}$};

    \fill[face] (-0.000,2.040) -- (2.043,-0.545) -- (-0.709,-1.410) -- cycle;
    \fill[facehighlight] (-0.000,2.040) -- (-1.333,-0.085) -- (-0.709,-1.410) -- cycle;
    \draw[edge] (-0.000,2.040) -- (-1.333,-0.085);
    \draw[edge] (-0.000,2.040) -- (2.043,-0.545);
    \draw[darkedge] (-0.000,2.040) -- (-0.709,-1.410);
    \draw[edge] (-1.333,-0.085) -- (-0.709,-1.410);
    \draw[edge] (2.043,-0.545) -- (-0.709,-1.410); 
\end{scope}
\end{tikzpicture}

\caption{Simplex states for $n=3$ (left) and $n=4$ (right), drawn as unit vectors in $\R^{n-1}$. 
Explicitly, $\ket{\psi_1}=\ket{0}$ and
$\ket{\psi_{2,3}}=-\frac{1}{2}\ket{0}\pm\frac{\sqrt{3}}{2}\ket{1}$ when $n = 3$, and
$\ket{\psi_i}=\frac{1}{\sqrt{3}}\left(\pm\ket{0}\pm\ket{1}\pm\ket{2}\right)$ with an
even number of minus signs when $n = 4$. 
Both multisets are $2$-learnable and saturate
Theorems~\ref{thm:main} and~\ref{thm:m_excl_threshold} at $k=2$.}
\label{fig:simplex-states}
\label{fig:trine-tetrahedral}
\end{figure}
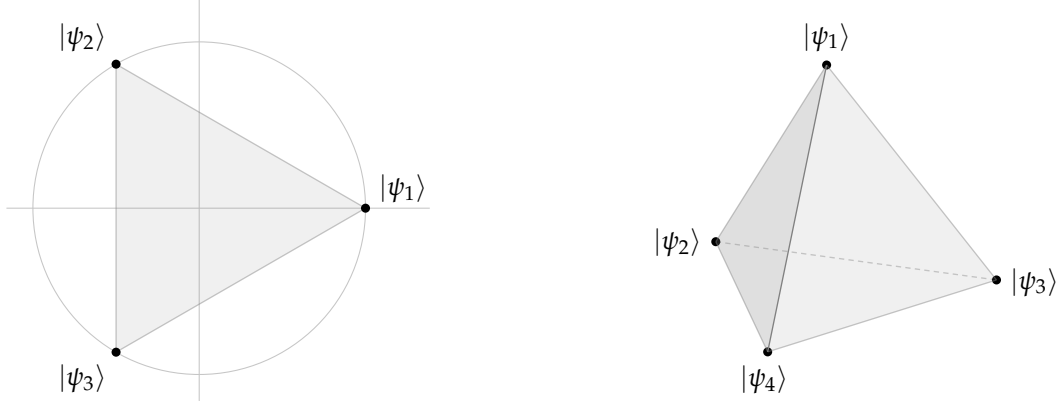 

We now state a result about $2$-learnability which is proved in Appendix~\ref{app:prop4}. 

\begin{proposition}\label{prop:row_regular_2_learning}
Let $n \geq 2$ be an integer and let $S=\big\{\ket{\psi_1}, \ket{\psi_2}, \ldots, \ket{\psi_n}\big\}$. 
For each $i \in \{1, 2,\ldots, n\}$, define
\begin{equation}
R_i
:=
\sum_{j=1}^n
\abs{\braket{\psi_i}{\psi_j}}.
\end{equation}
The following implications hold:
\begin{equation}
R_i \leq 2 \quad \text{for all $i$}\quad \Longrightarrow \quad\text{$S$ is $2$-learnable}\quad \Longrightarrow \quad \sum_{i=1}^n R_i \leq 2n.
\end{equation}
In particular, if $R:=R_i$ does not depend on $i$, then $S$ is $2$-learnable if and only if 
$
R\leq 2.
$
\end{proposition}   

In many examples of pure states that we consider, we have $R_1 = \cdots = R_n$ and thus the above proposition pins down exactly their $2$-learnability. 

\paragraph{Sharpness of the bounds in Theorems~\ref{thm:main} and \ref{thm:m_excl_threshold}  (Sections~\ref{sec:simplex} and~\ref{sec:sharp}).} 

While Theorems~\ref{thm:main} and \ref{thm:m_excl_threshold} are true, one may ask if these bounds are too aggressive, that is, if the right-hand side of the inequalities can be improved.  
We now discuss the optimality of these bounds.

Inequality~\eqref{eq:thm_main_frob} in  Theorem~\ref{thm:main} is saturated by an explicit family of rank-$(n-k+1)$ Fourier-type Gram matrices (see Section~\ref{sec:sharp}), and we exhibit an explicit $\vareps$-deformation of this family that ceases to be $k$-learnable for all $\vareps \in (0,1)$. 
The optimality of Inequality~\eqref{eq:thm_main_coh} is left as an open problem. 

Both bounds in Theorem~\ref{thm:m_excl_threshold} are optimal which can be seen by studying pure states with equal pairwise overlap. 
This parameterized multiset of pure states is discussed in detail in Section~\ref{sec:simplex} as well as why they are optimal for Theorem~\ref{thm:m_excl_threshold}.

\paragraph{Packing bounds (Section~\ref{sec:packing}).}  
Given that this new perfect classification task is closely tied to how close together or far apart the vectors are, this raises the question: 

\begin{quote} 
\textit{Given $S = \{\ket{\psi_1}, \ket{\psi_2}, \ldots, \ket{\psi_n} \} \subset \C^d$ which is $k$-learnable, what are the relationships\\ between the parameters $n$, $k$, and $d$? 
}
\end{quote} 

Such a bound can take the form of a \emph{packing inequality}; mathematically, it describes how many unit vectors forming a $k$-learnable multiset can be packed into a $d$-dimensional space. We show that 
\begin{equation} 
    n \le dk 
\end{equation} 
must hold in general. 
Moreover, this bound can be improved to 
\begin{equation}
    n \le \frac{d(k+1)}{2}
\end{equation} 
when no two states are global phases of one another. We prove that both of these bounds are tight via matching constructions.

\paragraph{Copy complexity (Section~\ref{sec:consequences}).}  

Intuitively, if instead of being given $\ket{\psi_i}$ to classify one is given $c$ copies of it, i.e.,  $\ket{\psi_i}^{\otimes c}$, then this makes the task of classification easier. 
Mathematically, taking additional copies decreases the magnitude of every overlap that lies strictly between zero and one.
In particular, we ask: 
\begin{quote} 
\textit{Given $S = \{\ket{\psi_1}, \ket{\psi_2}, \ldots, \ket{\psi_n} \}$ such that $\{\ket{\psi_1}^{\otimes c}, \ket{\psi_2}^{\otimes c}, \ldots, \ket{\psi_n}^{\otimes c} \}$ is $k$-learnable, what are the relationships between the parameters $n$, $k$, $c$, and $S$?   
}
\end{quote} 

One way to interpret this question comes from the \emph{copy complexity} perspective; what is the minimum $c$ (if it exists) such that a multiset $S$ becomes $k$-learnable?
We use Theorem~\ref{thm:main} and Proposition~\ref{prop:row_regular_2_learning} to bound the copy complexity of several well-known sets of states, notably:

\begin{theorem}[Informal]
    For prime qudit dimension $d$, the copy complexity of $2$-learning the set of $m$-qu$d$it stabilizer states is ${\Theta}(m)$.
\end{theorem}
See Corollary~\ref{cor:stab_qudit_exact} for a more precise statement. By using $m$ additional copies, one can learn an unknown $m$-qudit stabilizer state with failure probability at most $d^{-m}$ (see Remark~\ref{rmk:stabilizer_learning} for more details), resulting in an optimal $O(m)$ copy complexity for learning stabilizer states. This reproduces the information-theoretic $O(m)$ copy complexity obtained in~\cite{aaronson2008identifying,montanaro2017learning,allcock2025reconquering,CCLW26}, with the caveat that our result only applies to prime-dimensional qudits.  
In contrast to these works, our 2-learning measurement may not be efficiently implementable. Despite this drawback, an advantage to our result is that the 2-learning approach additionally yields a distinguished pair of stabilizer states guaranteed to contain the unknown state.

More generally, we upper bound the copy complexity for $k$-learning stabilizer states (Corollary~\ref{cor:stab}). We also give upper bounds on the copy complexity of $k$-learning arbitrary sets of states (Corollary~\ref{cor:copy_complexity}), SIC-POVMs (Corollary~\ref{cor:SIC_POVM}), and mutually unbiased bases (MUBs, Corollary~\ref{cor:MUB}). In the case of $2$-learning, we determine the copy complexity for SIC-POVMs and MUBs exactly (Corollaries ~\ref{cor:SIC_exact_2learning} and~\ref{cor:MUB_exact_2learning}). In both cases, the copy complexity of $2$-learning is at most $4$.

We note that our bounds are in stark contrast to the task of perfect distinguishability; a multiset of states that is not perfectly distinguishable will never become perfectly distinguishable by increasing the number of copies. For $k \geq 2$, however, a multiset that is not $k$-learnable may become $k$-learnable when sufficiently many copies are provided; Corollary~\ref{cor:copy_complexity} shows this happens whenever the states are pairwise distinct up to global phase. 

\paragraph{Application: Zero-error quantum state mutation detection problems (Section~\ref{sec:applications}).}     

There are many applications for which identifying quantum states is of central importance. 
As an application of our bounds, we revisit the task of anomaly
detection~\cite{RebentrostAnomaly, llorens2024quantum} and changepoint
detection~\cite{sentis2016quantum, sentis2017exact, mohan2023generalized}. In anomaly detection, we have states of the
form 
\begin{equation}
    \ket{\psi}^{\otimes {i-1}} \otimes \ket{\phi}\otimes \ket{\psi}^{\otimes n-i}
\end{equation} 
and the goal is to identify the location of the state $\ket{\phi}$ (called the anomaly). 
In changepoint detection, we have states of the form 
\begin{equation} 
\ket{\phi} \ket{\phi} \ket{\phi} \cdots \ket{\phi} \ket{\psi} \cdots \ket{\psi} 
\end{equation} 
and the goal is to identify the first location when the states changed (called the changepoint). 
In both settings, the goal is to locate a deviation from an otherwise regular sequence of states.  

In this paper, we study these problems via the following problem: 
\begin{quote}
\textit{When can we narrow the location of the anomaly or changepoint down to $k$ choices, and without error, as a function of the overlap $\braket{\psi}{\phi}$?}  
\end{quote} 

As an example result, we show that the multiset $\left\{\ket{\psi}^{\otimes {i-1}} \otimes \ket{\phi}\otimes \ket{\psi}^{\otimes n-i} : i \in \{1,2,\ldots,n\}\right\}$ is $k$-learnable if and only if
\begin{equation} 
        \abs{\braket{\phi}{\psi}} \leq \sqrt{\frac{k-1}{n-1}}.  
\end{equation} 
See Theorem~\ref{thm:single_anomaly} for a precise statement. Similar results, including a study of the \emph{multi-anomaly detection} problem~\cite{llorens2024quantum}, appear in Section~\ref{sec:applications}. 

\paragraph{Typical behaviour (Section~\ref{sec:typical}).}    

We now consider how learnable randomly-generated pure states are. 
Since randomly generated states in a large dimension are almost orthogonal, we expect a trade-off: as the dimension $d$ increases or the number of states $n$ decreases, learnability should increase, and as $d$ decreases or $n$ increases learnability should decrease. We start with a negative result in the regime where $d$ is small or $n$ is large:

\begin{proposition}\label{prop:random_gram_not_klearn_informalinformal}
    For fixed $\delta \in (0,1)$, independent Haar-random pure states $\ket{\psi_1}, \ket{\psi_2}, \ldots,\ket{\psi_n} \in \C^d$ are not $k$-learnable with probability at least $1-\delta$ whenever 
    \begin{equation}  
    d \leq C_{\delta} \min\{ n, (n/k)^2 \}
    \end{equation}
    for $k \in \{1, \ldots, n-1 \}$, where $C_{\delta}>0$ is a constant that depends only on $\delta$.
\end{proposition}

See Proposition~\ref{prop:random_gram_not_klearn} for a precise statement. 
Conversely, if $n$ is small or $d$ is large, the randomly generated pure states will be $k$-learnable with high probability:

\begin{proposition} 
    For fixed $\delta \in (0,1)$, independent Haar-random pure states $\ket{\psi_1}, \ket{\psi_2}, \ldots, \ket{\psi_n} \in \C^d$ are $k$-learnable with probability at least $1 - \delta$ whenever 
    \begin{equation}
        d \geq C_{\delta} (n(n - k + 1) / k) 
    \end{equation}
    for $k \in \{ 2, \ldots, n\}$, where $C_{\delta}>0$ is a constant that depends only on $\delta$. In particular, $d \leq C_{\delta} (n^2 / k)$ suffices when $k \leq n/2$.
\end{proposition}

See Proposition~\ref{prop:random_gram_markov} for a precise statement. 

\paragraph{Methods.}  
Operationally, $k$-learnability is a property of the existence of measurement strategies. 
Algebraically, however, it is a property of the Gram matrix. 
To specify this property, we define what it means for a matrix to be \emph{$k$-incoherent}. 

\begin{definition}[$k$-incoherence~\cite{ringbauer2018certification}] 
    Let $M$ be an $n \times n$ positive semidefinite matrix. We say that $M$ is \emph{$k$-incoherent} if it can be written as a sum of positive semidefinite matrices, each of which is supported only on a $k \times k$ principal submatrix. In other words, $M$ has a decomposition
    \begin{equation} 
        M = \sum_{\ell} M_{\ell}, 
    \end{equation} 
    where each $M_{\ell}$ is positive semidefinite and has $n-k$ rows that are identically $0$ (and similarly for those columns).
\end{definition}

\begin{remark}
    We note that the set of $k$-incoherent matrices is sometimes called the set of matrices with \emph{factor width} at most $k$ \cite{boman2005factor}. Also, the term ``$k$-incoherent'' is sometimes reserved for density matrices (i.e., positive semidefinite matrices with trace $1$); we do not make this restriction.
\end{remark} 

In the prior work \cite[Lemma 9]{johnston2025complexity}, we showed that determining the learnability of a multiset of pure states is exactly tied to the incoherence of its Gram matrix:  

\begin{proposition}\label{prop:klearnable}
    A multiset of pure states $S = \{\ket{\psi_1}, \ket{\psi_2}, \ldots, \ket{\psi_n}\} \subset \C^d$ is \emph{$k$-learnable} if and only if its Gram matrix is $k$-incoherent.
\end{proposition}  
Many of our results use this algebraic characterization to determine $k$-learnability. 

\section{Sharpness of \texorpdfstring{\cref{thm:m_excl_threshold}}{Theorem 3} via equiangular states} \label{sec:simplex} 

We now discuss a particular parameterized multiset of pure states which all have equal inner product with each other, called \emph{equiangular states}, and are defined implicitly via the $n \times n$ Gram matrix: 
   \begin{equation} \label{eq:1gamma}
        G_n(\gamma) = \begin{bmatrix}
            1 & \gamma & \gamma & \cdots & \gamma \\
            \gamma & 1 & \gamma & \cdots & \gamma \\
            \gamma & \gamma & 1 & \cdots & \gamma \\
            \vdots & \vdots & \vdots & \ddots & \vdots \\
            \gamma & \gamma & \gamma & \cdots & 1
        \end{bmatrix} 
    \end{equation} 
for $\gamma \in [0,1]$. 
Since it can be written as 
\begin{equation} 
G_n(\gamma) = 
\gamma \one\one^\t + (1 - \gamma) \I, 
\end{equation} 
where $\one$ is the all-ones column vector and $\I$ is the $n \times n$ identity matrix, we know that it is always positive semidefinite for this range of $\gamma$. 
Therefore, it is the Gram matrix of some collection of vectors and since the diagonal entries of $G_n(\gamma)$ are all equal to $1$, those vectors can be interpreted as pure states.  

We state a useful lemma about equiangular states.  

\begin{lemma}\label{lem:useful}
    For $n \geq 2$, $k \in \{ 1, 2, \ldots, n \}$, we have that $G_n(\gamma)$ is $k$-incoherent when $0 \leq \gamma \leq \frac{k-1}{n-1}$. 
\end{lemma} 

\begin{proof} 
    This proof is similar to one in ~\cite{johnston2023antidistinguishability} which proved it for the case of $k= n - 1$, corresponding to antidistinguishability. 
    The case of $k = n$ is trivial as every $n \times n$ positive semidefinite matrix is $n$-incoherent. 
    The case of $k = 1$ is also trivial as the condition enforces $\gamma = 0$, i.e., a diagonal matrix, which is always $1$-incoherent.
    The rest of the proof proves the case when $k \in \{ 2, \ldots, n-1 \}$ and, by extension, $n \geq 3$.   
    
    Set $x = (k-1)/(n-1)$ for notational clarity and, for each $I \subseteq \{1, 2, \ldots, n\}$ with $\abs{I} = n-k$, define
    \begin{equation} 
        G_I := \frac{1}{\binom{n-2}{n-k}} 
        \left(
        (x - \gamma) ( \I - \sum_{i \in I} \e_i\e_i^\t ) + \gamma\left(\one-\sum_{i \in I}\e_i\right)\left(\one-\sum_{i \in I}\e_i\right)^\t
        \right). 
    \end{equation} 
    It is clear that $G_I$ is positive semidefinite (since $\gamma \leq (k-1)/(n-1) = x$) and is equal to $0$ outside of the $k \times k$ principal submatrix that avoids rows and columns with indices in $I$. Furthermore, direct calculation shows that
    \begin{align}\begin{split}\label{eq:G_gamma_decomp}
        \sum_{\substack{I \subseteq \{1, 2, \ldots, n\} \\ \abs{I} = n-k}} G_I & = \frac{1}{\binom{n-2}{n-k}}\sum_{\substack{I \subseteq \{1, 2, \ldots, n\} \\ \abs{I} = n-k}} \left((x-\gamma) (\I - \sum_{i \in I}\e_i\e_i^\t\Big) + \gamma\left(\one-\sum_{i \in I}\e_i\right)\left(\one-\sum_{i \in I}\e_i\right)^\t\right) \\
        & = \frac{1}{\binom{n-2}{n-k}}\left((x-\gamma)\binom{n-1}{n-k}\I + \gamma\binom{n-1}{n-k}\left(x\one\one^\t + \left(1-x\right)\I\right)\right) \\
        & = \frac{x\binom{n-1}{n-k}}{\binom{n-2}{n-k}}\left((1-\gamma)\I + \gamma\one\one^\t\right) \\
        & = G_n(\gamma),
    \end{split}\end{align}
    so it follows that $G_n(\gamma)$ is $k$-incoherent by definition.   
\end{proof}

By~\cref{prop:klearnable}, we know that the equiangular states with $0 \leq \gamma \leq \frac{k-1}{n-1}$ must be $k$-learnable. 
On the other hand, we have 
\begin{equation} 
\sum_{i,j=1}^n \abs{\braket{\psi_i}{\psi_j}} = \sum_{i,j=1}^n \abs{G_{i,j}} = n + n(n-1) \gamma.  
\end{equation} 
Thus by \cref{thm:m_excl_threshold}, if $\gamma > \frac{k-1}{n-1}$, then it cannot be $k$-learnable. 
Combining, we have the following corollary. 

\begin{corollary}\label{cor:equi_angle}
    Let $0 \leq \gamma \leq 1$, $n \geq 2$,  $k \in \{1, 2, \ldots, n \}$, and let $S = \{\ket{\psi_1},
    \ket{\psi_2}, \ldots, \ket{\psi_n}\}$ be such that $\braket{\psi_i}{\psi_j} = \gamma$ whenever $i \neq j$. 
    In other words, $S$ has Gram matrix given by Equation~\cref{eq:1gamma}. 
    Then $S$ is $k$-learnable if and only if
    \begin{equation}
        \gamma \leq \frac{k-1}{n-1}.
    \end{equation} 
\end{corollary} 

By~\cref{cor:equi_angle}, setting $\gamma = \frac{k-1}{n-1}$ yields a set of $n$ pure states that are $k$-learnable, proving that both bounds in \cref{thm:m_excl_threshold} are sharp. 

\section{Sharpness of the Frobenius norm bound of \texorpdfstring{Inequality~\eqref{eq:thm_main_frob}}{Inequality (3)}} \label{sec:sharp}

We now show that Inequality~\eqref{eq:thm_main_frob} of Theorem~\ref{thm:main} is tight in the sense that the right-hand side of the inequality cannot be increased at all. To this end, note that
\[
    \sum_{i, j = 1}^{n} \abs{\braket{\psi_i}{\psi_j}}^2 = \norm{G}_{\textup{F}}^2,
\]
where $G$ is the Gram matrix of $\ket{\psi_1}$, $\ket{\psi_2}$, $\ldots$, $\ket{\psi_n}$ and $\norm{\cdot}_{\textup{F}}$ denotes the Frobenius norm. We prove tightness by demonstrating how to construct a Gram matrix that is not $k$-incoherent (and thus corresponds to a set of states that is not $k$-learnable) with any Frobenius norm that is slightly larger than $n/\sqrt{n-k+1}$.

\begin{theorem}\label{thm:main_tight}
    Let $1 < k < n$ be integers, let $r := n-k+1$, and let $\omega := e^{\frac{2\pi i}{n}}$. For each $j \in \{1,2,\ldots,r\}$, let
    \begin{align}
        \vv_j = \left(1, \omega^{j-1}, \omega^{2(j-1)},\ldots, \omega^{(n-1)(j-1)}\right) \in \C^n.
    \end{align}
    Consider the Gram matrix
    \begin{align}
        G := \frac{1}{r} \sum_{j=1}^r \vv_j \vv_j^*.
    \end{align}
    Then $G$ has $\norm{G}_{\textup{F}} = n/\sqrt{r}$ and is $k$-incoherent but not $(k-1)$-incoherent. Furthermore, for any $\vareps \in (0,1)$, the Gram matrix
    \begin{align}
        G_{\vareps} := \frac{1}{r} \left((1+\vareps) \vv_1 \vv_1^* + (1-\vareps) \vv_r \vv_r^* + \sum_{i=2}^{r-1} \vv_i \vv_i^* \right)
    \end{align}
    is not $k$-incoherent.
\end{theorem}

In particular, since
\begin{equation}
    \lim_{\vareps \rightarrow 0^{+}} G_{\vareps} = G \quad \text{and} \quad \lim_{\vareps \rightarrow 0^{+}} \norm{G_{\vareps}}_{\textup{F}} = \lim_{\vareps \rightarrow 0^{+}} \left(\frac{n}{\sqrt{r}}\sqrt{1 + \frac{2\vareps^2}{r}}\right) = \frac{n}{\sqrt{r}},
\end{equation}
Theorem~\ref{thm:main_tight} shows that Inequality~\eqref{eq:thm_main_frob} is tight: for any Frobenius norm slightly above $n/\sqrt{r}$, we can choose $\vareps > 0$ so as to create a Gram matrix $G_{\vareps}$ with that Frobenius norm that is not $k$-incoherent, so (by Proposition~\ref{prop:klearnable}) the associated set of states is not $k$-learnable.

\begin{proof}[Proof of Theorem~\ref{thm:main_tight}]
    The fact that $G$ and $G_{\vareps}$ are Gram matrices is straightforward to verify. Since the vectors $\vv_j$ are mutually orthogonal, the Gram matrix $G$ is proportional to a rank-$r$ orthogonal projection, and it is easily verified that $\norm{G}_{\textup{F}} = n/\sqrt{r}$. By Theorem~\ref{thm:main}, $G$ is $k$-incoherent.
    
    Let $\y \in \im(G)$ (the image of $G$) be a vector with at most $n - r = k - 1$ nonzero entries, and let $j_1$, $\ldots$, $j_r$ be indices for which $y_{j_1} = \cdots = y_{j_r} = 0$. Since $\vv_1, \ldots, \vv_r$ form a basis of $\im(G)$, there exist scalars $\{\alpha_j\} \subset \C$ such that $\y = \sum_{j=1}^r \alpha_j \vv_j$. It follows that
    \begin{equation}\label{eq:lin_comb_root_unity}
        \sum_{j=1}^r \alpha_j \omega^{(j_{\ell}-1)(j-1)} = 0
    \end{equation}
    for all $\ell \in \{1,2,\ldots,r\}$. If we define the polynomial $p(x) = \sum_{j=1}^r \alpha_j x^{j-1}$ then Equation~\eqref{eq:lin_comb_root_unity} says exactly that
    \begin{equation}
        p(\omega^{j_1-1}) = \cdots = p(\omega^{j_r-1}) = 0.
    \end{equation}
    Since $p$ is a polynomial of degree at most $r-1$ and has $r$ roots, it must be the case that $p(x) = 0$ for all $x \in \C$. It follows that $\alpha_1 = \cdots = \alpha_r = 0$, so $\y = \mathbf{0}$. We have thus shown that there are no nonzero vectors in $\im(G)$ with at most $k-1$ nonzero entries, so $G$ is not $(k-1)$-incoherent.
    
    It remains only to prove that $G_{\vareps}$ is not $k$-incoherent. Similar to before, let $\y \in \im(G_{\vareps})$ have at most $k = n - r + 1$ nonzero entries, and let $j_1$, $\ldots$, $j_{r-1}$ be indices for which $y_{j_1} = \cdots = y_{j_{r-1}} = 0$. Then there exist scalars $\{\alpha_j\} \subset \C$ (in particular, $\alpha_j = \vv_j^*\y/n$) such that $\y = \sum_{j=1}^r \alpha_j \vv_j$, so
    \begin{equation}\label{eq:lin_comb_root_unityb}
        \sum_{j=1}^r \alpha_j \omega^{(j_{\ell}-1)(j-1)} = 0
    \end{equation}
    for all $\ell \in \{1,2,\ldots,r-1\}$.  If we define the polynomial $p(x) = \sum_{j=1}^r \alpha_j x^{j-1}$ then Equation~\eqref{eq:lin_comb_root_unityb} says exactly that
    \begin{equation}
        p(\omega^{j_1-1}) = \cdots = p(\omega^{j_{r-1}-1}) = 0.
    \end{equation}
    Since $p$ is a polynomial of degree at most $r-1$, and we have found $r-1$ roots of it, it must have the form $p(x) = \alpha_r (x-\omega^{j_1-1}) \cdots (x-\omega^{j_{r-1}-1})$. It follows that $\alpha_1=\pm \alpha_r (\omega^{j_1-1}\cdots \omega^{j_{r-1}-1}),$ so $\abs{\alpha_1} = \abs{\alpha_r}$, which implies
    \begin{equation}\label{eq:equal_abs_value_tight}
        \abs{\vv_1^*\y} = \abs{\vv_r^*\y}.
    \end{equation}
    
    Now suppose (for the sake of establishing a contradiction) that $G_{\vareps}$ is $k$-incoherent, so that it has a decomposition of the form $G_{\vareps} = \sum_{i} \y_i \y_i^{*}$, where each $\y_i$ has at most $k$ nonzero entries. Since each $\vv_j$ is an eigenvector of $G_{\vareps}$, it follows that the associated eigenvalue $\lambda_j$ is given by
    \begin{equation}\label{eq:sum_of_ip_evs_tight}
        \lambda_j = \frac{1}{n}\vv_j^*G_{\vareps}\vv_j = \frac{1}{n}\sum_{i} \abs{\vv_j^*\y_i}^2.
    \end{equation}
    Notice that the form of $G_{\vareps}$ given in the statement of the theorem tells us that if we sort the eigenvalues non-increasingly then we have $\lambda_1 = n(1+\vareps)/r$, $\lambda_j = n/r$ for $j \in \{2,3,\ldots,r-1\}$, and $\lambda_r = n(1-\vareps)/r$. Since $\y_i \in \im(G_{\vareps})$ for all $i$, Equation~\eqref{eq:equal_abs_value_tight} tells us that $\abs{\vv_1^*\y_i} = \abs{\vv_r^*\y_i}$ for all $i$. But Equation~\eqref{eq:sum_of_ip_evs_tight} then implies
    \begin{equation}
        \frac{n}{r}(1+\vareps) = \lambda_1 = \lambda_r = \frac{n}{r}(1-\vareps),
    \end{equation}
    which is a contradiction that completes the proof.
\end{proof}

\begin{remark}
    The proof of Theorem~\ref{thm:main_tight} provides a way of constructing a Gram matrix that is not $k$-incoherent and whose spectrum is arbitrarily close to
    \begin{equation}\label{eq:flat_spectrum}
        \frac{n}{n-k+1}, \frac{n}{n-k+1}, \ldots, \frac{n}{n-k+1}, 0, 0, \ldots, 0,
    \end{equation}
    where the eigenvalue $n/(n-k+1)$ occurs $n-k+1$ times. In the terminology of \cite{johnston2022absolutely}, the same argument shows that~\eqref{eq:flat_spectrum} is the only spectrum belonging to an absolutely $k$-incoherent matrix of trace $n$ and rank at most $n-k+1$. Theorem~\ref{thm:main_tight} can thus be regarded as a strengthening of \cite[Theorem~4]{johnston2022absolutely}, which showed the weaker statement that no matrix with rank strictly less than $n-k+1$ is absolutely $k$-incoherent.
\end{remark}

\section{Packing bounds for \texorpdfstring{$k$}{k}-learnability}\label{sec:packing}

When $k = 1$, $k$-learnability of a multiset of states is equivalent to those states being mutually orthogonal. As a result, the number of states $n$ in a $1$-learnable set cannot exceed the dimension $d$. The following theorem generalizes this observation to higher values of $k$, establishing bounds on the number of states in a $k$-learnable multiset in terms of the dimension.

\newpage 

\begin{theorem}\label{thm:packing}
    Let $\big\{\ket{\psi_1}, \ket{\psi_2}, \ldots, \ket{\psi_n}\big\} \subset \C^d$ be a multiset of pure states that is $k$-learnable. 
    \begin{enumerate}
        \item[(1)] It holds that $n \leq dk$.
        \item[(2)] If the vectors are pairwise non-parallel, i.e., $\ket{\psi_i} \notin \spn\{\ket{\psi_j}\}$ for all $i \neq j$, then $2n \leq d(k+1)$.
    \end{enumerate}
\end{theorem}

\begin{proof}
    We may assume that $\dim(\spn\{\ket{\psi_1}, \ket{\psi_2}, \ldots, \ket{\psi_n}\}) = d$, as this can only make the desired inequality tighter. Let $G$ be the Gram matrix of this multiset, and let $\vv_1, \vv_2, \ldots, \vv_m \in \C^n$ be vectors with at most $k$ nonzero entries such that
    \begin{equation}
        G = \sum_{j=1}^m \vv_j \vv_j^*.
    \end{equation}
    It follows that $\im(G)$ (the image of $G$) is spanned by $\vv_1, \vv_2, \ldots, \vv_m$. After permuting terms, we may assume that $\{\vv_1, \vv_2, \ldots, \vv_d\}$ forms a basis for $\im(G)$. Let $A = [\vv_1, \vv_2, \ldots, \vv_d] \in \C^{n \times d}$ and $Y = [\ket{\psi_1}, \ket{\psi_2}, \ldots, \ket{\psi_n}]^* \in \C^{n \times d}$. Since $\im(A) = \im(G) = \im(Y)$, there exists an invertible matrix $C \in \C^{d \times d}$ for which $A = YC$.
    
    We first prove the bound $n \leq dk$. Since the vectors $\ket{\psi_j}$ are nonzero, it follows that the rows of $A$ are nonzero. Let $s$ denote the total number of nonzero entries of $A$. Since every row is nonzero, we have $s \geq n$. On the other hand, $A$ has $d$ columns and each column has support size at most $k$, so $s \leq kd$. Combining the two bounds gives $n \leq kd$.
    
    Now we prove the bound $2n \leq d(k+1)$ under the assumption that the vectors are pairwise non-parallel. Note that the rows of $Y$ are pairwise non-parallel, as these are the vectors $\ket{\psi_j}^*$. Since $C$ is invertible, the rows of $A$ are also pairwise non-parallel. Let $t$ be the number of rows of $A$ having exactly one nonzero entry. Any such row is a nonzero scalar multiple of one of the standard basis vectors $\e_1, \e_2, \ldots, \e_d$. Since the rows are pairwise non-parallel, there can be at most one such row for each coordinate, so $t \leq d$. Let $s$ denote the total number of nonzero entries of $A$. Since every row is nonzero, and every row that is not a scalar multiple of an $\e_j$ has support size at least $2$, we have
    \begin{equation}
        s \geq t + 2(n-t) = 2n-t \geq 2n-d.
    \end{equation}
    On the other hand, $A$ has $d$ columns and each column has support size at most $k$, so $s \leq kd$. Combining the two bounds gives $2n-d \leq kd$, hence $2n\leq d(k+1)$, which completes the proof.
\end{proof} 

Next, we show that the bounds of Theorem~\ref{thm:packing} are tight: for any tuple of integers $(k,d,n)$ satisfying the inequalities specified by that theorem, there exists a $k$-learnable multiset of $n$ pure states in $\C^d$.

\begin{theorem}\label{thm:packing_tight}
    Let $k, d \leq n$ be positive integers. 
    \begin{enumerate}
        \item[(1)] If $n \leq dk$, then there exists a multiset of pure states $\big\{\ket{\psi_1}, \ket{\psi_2}, \ldots,\ket{\psi_n}\big\} \subset \C^d$ that is $k$-learnable.
        \item[(2)] If $2n\leq d(k+1)$, then there exists a $k$-learnable multiset of pure states $\big\{\ket{\psi_1}, \ket{\psi_2}, \ldots, \ket{\psi_n}\big\} \subset \C^d$ that is pairwise non-parallel, i.e., $\ket{\psi_i} \notin \spn\{\ket{\psi_j}\}$ for all $i \neq j$.
    \end{enumerate}
\end{theorem}

\begin{proof}
    We begin by proving the first statement. Write $n=qk+s$ for a positive integer $q$ and $0 \leq s < k$.
    Consider the following $n$ unit vectors:
    \begin{equation}
        \underbrace{\e_1,\ldots,\e_1}_{k\text{ times}}, \
        \underbrace{\e_2,\ldots,\e_2}_{k\text{ times}}, \
        \ldots, \
        \underbrace{\e_q,\ldots,\e_q}_{k\text{ times}}, \
        \underbrace{\e_{q+1},\ldots,\e_{q+1}}_{s\text{ times}}.
    \end{equation}
    These vectors span a space of dimension at most $d$, and are clearly $k$-learnable.
    
    Now we prove the second statement. Let $m = n-d$, so $2m = 2n-2d \leq d(k-1)$. Let $H$ be a multigraph on vertex set $\{1,2,\ldots,d\}$ with $m$ edges and maximum vertex degree at most $k-1$. To construct $H$, start with the edgeless graph on $d$ vertices and at each step add an edge between any two vertices whose degrees are at most the degrees of all of the remaining $d-2$ vertices. One can easily check that the difference between the maximum and minimum vertex degree at any step is at most 1. Furthermore, the inequality $2m \leq d(k-1)$ guarantees that if $m$ edges are added in this way, no vertex will have degree greater than $k-1$.
    
    Let $A \in \C^{n \times d}$ be defined as follows. First include the $d$ standard basis vectors $\e_1, \e_2, \ldots, \e_d$ as rows. Next, list the edges of $H$ as $\{i_t,j_t\}$ for $t \in \{1,2,\ldots,m\}$. Choose phases $e^{i \phi_t}$ (distinct if parallel edges occur) and define additional rows of $A$ to be
    \begin{equation}
        \w_t := \frac{1}{\sqrt{2}}\left(\e_{i_t}+e^{i\phi_t} \e_{j_t}\right) \quad \text{for} \quad t \in \{1, 2, \ldots, m\}.
    \end{equation}
    It follows that $A$ has $d+m = n$ rows, and every row has norm $1$, so $G := AA^*$ is a Gram matrix of $n$ unit vectors. The first $d$ rows of $A$ are $\e_1, \e_2, \ldots, \e_d$, so $\rank(A) = d$, which implies $\rank(G) = \rank(AA^*) = d$.
    
    Let $\vv_j$ denote the $j$-th column of $A$. Then
    \begin{equation}
        G = AA^* = \sum_{j=1}^d \vv_j\vv_j^* .
    \end{equation}
    Notice that $\vv_j$ is nonzero in its $j$-th entry (since that entry comes from the row $\e_j$ in $A$) and in the entries corresponding to edges incident to vertex $j$. Hence
    \begin{equation}
        \abs{\supp(\vv_j)} = 1 + \deg_H(j) \leq k.
    \end{equation}
    In particular, this implies that $G$ is $k$-incoherent. Hence, the complex conjugates of the rows of $A$ form a $k$-learnable multiset of pure states, which completes the proof.
\end{proof} 

\section{Copy complexity bounds for \texorpdfstring{$k$}{k}-learnability}\label{sec:consequences}

The number of copies of an unknown quantum state required to perform a learning task is commonly called its \emph{sample complexity} or \emph{copy complexity}. It has been studied for problems including state tomography~\cite{HaahEtAl2017Tomography,montanaro2017learning}, PAC and shadow learning~\cite{Aaronson2007Learnability,Aaronson2018ShadowTomography,HuangKuengPreskill2020}, state certification~\cite{BadescuODonnellWright2019}, and property testing~\cite{ODonnellWright2021QuantumSpectrumTesting}. Here we study an exact, zero-error version of copy complexity: given a finite multiset $S$, we ask for the minimum number of copies after which every measurement outcome identifies a list of at most $k$ candidates containing the unknown state. Since
\[
\braket{\psi_i^{\otimes c}}{\psi_j^{\otimes c}}
=\braket{\psi_i}{\psi_j}^{\,c},
\]
taking additional copies suppresses nontrivial overlaps, and we determine when this makes the resulting states $k$-learnable.

\begin{definition}
    Let $S=\{\ket{\psi_1}, \ket{\psi_2}, \ldots,\ket{\psi_n}\} \subset \C^d$ be a multiset of pure states. We say that $S$ is \emph{$(c,k)$-learnable} if the multiset $S^{(c)}:=\{\ket{\psi_1}^{\otimes c},\ldots,\ket{\psi_n}^{\otimes c}\}$ is $k$-learnable.
    
    We define the \emph{copy complexity of $k$-learning} for $S$, denoted $c_k(S)$, to be the minimum $c$ for which $S$ is $(c,k)$-learnable. If no finite $c$ exists, then we define $c_k(S)=\infty$.
\end{definition}

\subsection{Copy complexity bounds from maximum pairwise overlap}

Our main theorem provides a generic copy complexity bound for $2$-learning an arbitrary multiset in terms of the overlap structure of the multiset of states.

\begin{corollary}\label{cor:copy_complexity}
    Let $S = \{\ket{\psi_1}, \ket{\psi_2}, \ldots, \ket{\psi_n}\} \subset \C^d$ be a multiset of $n \geq 3$ pure states with maximum pairwise overlap $0 < \delta := \max_{i \neq j} \abs{\braket{\psi_i}{\psi_j}}^2 < 1$. Then $S$ is $(c, 2)$-learnable for
    \begin{equation}
        c = \ceil{\frac{2\log(n-1)}{\log(1/\delta)}}.
    \end{equation}
\end{corollary}

\begin{proof}
    The squared Frobenius norm of the Gram matrix $G^{(c)}$ of $S^{(c)} =
    \{\ket{\psi_i}^{\otimes c}\}$ satisfies
    \begin{equation}
        \norm{{G^{(c)}}}_{\textup{F}}^2 = \sum_{i,j=1}^n \abs{\braket{\psi_i}{\psi_j}}^{2c}
        = n + \sum_{i \neq j} \abs{\braket{\psi_i}{\psi_j}}^{2c}
        \leq n + n(n-1)\delta^c.
    \end{equation}
    By Theorem~\ref{thm:main} with $k=2$, $S$ is $(c,2)$-learnable
    whenever $\norm{{G^{(c)}}}_{\textup{F}}^2 \leq n^2/(n-1)$. This holds when
    \begin{equation}
        n(n-1)\delta^c \leq \frac{n}{n-1},
    \end{equation}
    i.e., $\delta^c \leq (n-1)^{-2}$. Taking logarithms gives $c \geq
    2\log(n-1)/\log(1/\delta)$.
\end{proof}

\subsection{SIC-POVMs and mutually unbiased bases}

In this section we determine copy complexity bounds for $k$-learnability of SIC-POVMs and mutually unbiased bases. 
We use a straightforward result on $(c,k)$-learnability of complex projective $c$-designs, which we soon state.

We say that a multiset of states $S = \big\{\ket{\psi_1}, \ket{\psi_2}, \ldots, \ket{\psi_n}\big\} \subset \C^d$ forms a \emph{(complex projective) $c$-design} if
\begin{equation}
    \sum_{i=1}^n \ketbra{\psi_i}{\psi_i}^{\otimes c} = \frac{n}{\binom{d+c-1}{c}} \Pi_{d,c},
\end{equation}
where $\Pi_{d,c}$ is the projection to $S^c(\C^d) \subseteq (\C^d)^{\otimes c}$. Here, $S^c(\C^d)$  is the \emph{symmetric subspace} of vectors ${\bf v} \in (\C^d)^{\otimes c}$ for which $v_{i_1,\dots,i_c}=v_{i_{\sigma(1)},\dots,i_{\sigma(c)}}$ for all permutations $\sigma \in \mathfrak{S}_c$.

\begin{corollary}\label{cor:learning_designs}
    Let $S = \big\{\ket{\psi_1}, \ket{\psi_2}, \ldots, \ket{\psi_n}\big\} \subset \C^d$ be a multiset of pure states that forms a $c$-design. Then $S$ is $(c,k)$-learnable for $k = n-\binom{d+c-1}{c}+1$.
\end{corollary}

\begin{proof}
    Let $A = [\ket{\psi_1}^{\otimes c}, \ldots, \ket{\psi_n}^{\otimes c}]$ be the $d^c \times n$ matrix with columns $\ket{\psi_i}^{\otimes c}$. By the definition of $c$-design, it holds that
    \begin{equation}
        AA^* = \sum_{i=1}^n \ketbra{\psi_i}{\psi_i}^{\otimes c} = \frac{n}{\binom{d+c-1}{c}} \Pi_{d,c}.
    \end{equation}
    Note that $G=A^*A$ has the same nonzero eigenvalues as $AA^*$, which are $n \binom{d+c-1}{c}^{-1}$ with multiplicity $\binom{d+c-1}{c}$. It follows that 
    \begin{equation}
        \norm{G}_{\textup{F}} = \frac{n}{\sqrt{\binom{d+c-1}{c}}}.
    \end{equation}
    The corollary now follows immediately from Theorem~\ref{thm:main}.
\end{proof}

\subsubsection{SIC-POVMs}

A \emph{symmetric informationally complete POVM} (SIC-POVM) is a set $S = \{\ket{\psi_1}, \ket{\psi_2}, \ldots,\ket{\psi_{d^2}}\} \subset \C^d$ for which $\abs{\braket{\psi_i}{\psi_j}}^2 = 1/(d+1)$ for all $i \neq j$. We first use Proposition~\ref{prop:row_regular_2_learning} to determine the copy complexity of $2$-learning a SIC-POVM.

\begin{corollary}\label{cor:SIC_exact_2learning}
    Let $S = \big\{\ket{\psi_1}, \ket{\psi_2}, \ldots,\ket{\psi_{d^2}}\big\} \subset \C^d$ be a SIC-POVM. Then
    \begin{equation}
        c_2(S) =
        \begin{cases}
            2, & d = 2,\\
            3, & d = 3,\\
            4, & d \geq 4.
        \end{cases}
    \end{equation}
\end{corollary}

\begin{proof}
    For a SIC-POVM, $\abs{\braket{\psi_i}{\psi_j}}^2 = 1/(d+1)$ whenever $i \neq j$. Hence the Gram matrix of $S^{(c)} = \big\{\ket{\psi_1}^{\otimes c}, \ldots, \ket{\psi_{d^2}}^{\otimes c}\big\}$ satisfies
    \begin{equation}
        R := \sum_{j=1}^{d^2} \abs{\braket{\psi_i}{\psi_j}}^c = 1 + \frac{d^2-1}{\sqrt{(d+1)^c}}
    \end{equation}
    for all $i \in \{1, 2, \ldots, d^2\}$. By Proposition~\ref{prop:row_regular_2_learning}, $S$ is $(c,2)$-learnable if and only if $R\leq2$, or equivalently,
    \begin{equation}
        c \geq \frac{2\log(d^2-1)}{\log(d+1)}.
    \end{equation}
    Therefore
    \begin{equation}
        c_2(S) = \ceil{\frac{2\log(d^2-1)}{\log(d+1)}}.
    \end{equation}
    One can easily verify that this simplifies to the expression in the corollary statement.
\end{proof}

The above corollary shows that $4$ copies always suffice to $2$-learn a SIC-POVM. We now use Theorem~\ref{thm:main} to provide a bound on the $k$-learnability of a SIC-POVM if we are given fewer than $4$ copies.

\begin{corollary}\label{cor:SIC_POVM}
    Let $S \subset \C^d$  be a SIC-POVM. Then $S$ is
    \begin{align}
        &\text{$k$-learnable}\quad \text{ for} \quad  k=d^2-d+1,\\
        &\text{$(2,k)$-learnable}\quad\text{ for} \quad k=\binom{d}{2}+1,\\
        &\text{$(3,k)$-learnable}\quad\text{ for} \quad  k=\ceil{\frac{d^2+3}{d+3}}.
     \end{align}
\end{corollary}
\begin{proof}
    It is known that SIC-POVMs form $1$-designs and $2$-designs. Hence, $S$ is $k$-learnable for $k=d^2-d+1$, and $(2,k)$-learnable for $k=\binom{d}{2}+1$ by Corollary~\ref{cor:learning_designs}. In general, for a SIC-POVM $S$, the Gram matrix $G^{(c)}$ of $S^{(c)}$ satisfies
    \begin{equation}
        \norm{{G^{(c)}}}_{\textup{F}}^2 = d^2 + \frac{d^4-d^2}{(d+1)^c}.
    \end{equation}
    By Theorem~\ref{thm:main}, $S$ is $(c,k)$-learnable if
    \begin{equation}
        d^2 + \frac{d^4-d^2}{(d+1)^c} \leq \frac{d^4}{d^2-k+1}.
    \end{equation}
    Rearranging, we have that $S$ is $(c,k)$-learnable for
    \begin{equation}
        k = 1 + \ceil{\frac{d^2(d^2-1)}{(d+1)^c+d^2-1}}.
    \end{equation}
    For $c=3$, this gives
    \begin{equation}
        k = \ceil{\frac{d^2+3}{d+3}}.
    \end{equation}
    This completes the proof.
\end{proof}

\subsubsection{Mutually unbiased bases}

A collection of \emph{mutually unbiased bases} (MUBs) is a set
\begin{equation}
    S = \{\ket{\psi_{i,j}} : i \in \{1,2,\ldots,d\}, j \in \{1,2,\ldots,s\}\} \subset \C^d
\end{equation}
for some integer $s \geq 2$ such that $\{\ket{\psi_{i,j}} : i \in \{1,2,\ldots,d\}\}$ is an orthonormal basis of $\C^d$ for each $j$, and $\abs{\braket{\psi_{i,j}}{\psi_{i',j'}}}^2 = 1/d$ for all $i,i'$ and $j \neq j'$. It is known that $s\leq d+1$ must hold for any collection of MUBs.

We first use Proposition~\ref{prop:row_regular_2_learning} to determine the exact copy complexity of $2$-learning a collection of MUBs.

\begin{corollary}\label{cor:MUB_exact_2learning}
    Let $S = \left\{\ket{\psi_{i,j}} : i \in \{1,2,\ldots,d\}, j \in \{1,2,\ldots,s\}\right\} \subset \C^d$ be a collection of $s \geq 2$ mutually unbiased bases. Then
    \begin{equation}
        c_2(S) = \ceil{2+2\log_d(s-1)}.
    \end{equation}
    In particular, $c_2(S) \leq 4$, with equality for a complete collection of $d+1$ mutually unbiased bases.
\end{corollary}

\begin{proof}
    Fix $\ket{\psi_{i,j}} \in S$. The states in the same basis are orthogonal, while for every $j' \neq j$ and every $i' \in \{1,2,\ldots,d\}$, we have
    \begin{equation}
        \abs{\braket{\psi_{i,j}}{\psi_{i',j'}}} = \frac{1}{\sqrt{d}}.
    \end{equation}
    It follows that the sum of the absolute values of any row in the Gram matrix of
    \begin{equation}
        S^{(c)} = \left\{\ket{\psi_{i,j}}^{\otimes c} : i \in \{1,2,\ldots,d\}, j \in \{1,2,\ldots,s\} \right\}
    \end{equation}
    is given by $R = 1+(s-1)d^{1-c/2}$. By Proposition~\ref{prop:row_regular_2_learning}, $S$ is $(c,2)$-learnable if and only if $R \leq2 $, or equivalently if and only if
    \begin{equation}
        c \geq 2+2\log_d(s-1).
    \end{equation}
    This proves the general expression for $c_2(S)$.

    Finally, since $s \leq d+1$ for any collection of mutually unbiased bases, $\log_d(s-1) \leq 1$, so $c_2(S) \leq 4$. If $s = d+1$, then $c_2(S) = \ceil{2+2\log_d d} = 4$, so the bound is attained for a maximal collection of MUBs.
\end{proof}

The previous corollary showed that $4$ copies of a collection of MUBs always suffice to $2$-learn them. The following corollary tells us what (weaker) type of learning is possible if we have fewer than $4$ copies.

\begin{corollary}\label{cor:MUB}
    Let $S = \left\{\ket{\psi_{i,j}} : i \in \{1,2,\ldots,d\}, j \in \{1,2,\ldots,s\}\right\} \subset \C^d$ be a collection of $s \geq 2$ mutually unbiased bases. Then $S$ is
    \begin{align}
        \text{$k$-learnable}\quad \text{ for} \quad & k=sd-d+1,\\
        \text{$(2,k)$-learnable}\quad\text{ for} \quad & k=1+{\ceil{\frac{s(s-1)d}{d+s-1}},}\\
        \text{$(3,k)$-learnable}\quad\text{ for} \quad & k={1+ \ceil{\frac{s(s-1)d}{d^2+s-1}}.}
            \end{align}
\end{corollary}

\begin{proof}
    Let $G^{(c)}$ be the Gram matrix of $S^{(c)}$. By direct calculation,
    \begin{align}
        \norm{{G^{(c)}}}_{\textup{F}}^2 = sd + \frac{s(s-1)d^2}{d^c}.
    \end{align}
    By Theorem~\ref{thm:main}, $S$ is $(c,k)$-learnable if 
    \begin{equation}
        \norm{{G^{(c)}}}_{\textup{F}}^2 \leq \frac{(sd)^2}{sd-k+1}.
    \end{equation}
    Rearranging, it is sufficient to set
    \begin{equation}
        k = 1+\ \ceil{\frac{s(s-1)d}{d^{c-1}+s-1}}.
    \end{equation}
    This reproduces the above expressions, and completes the proof.
\end{proof}

\subsection{Stabilizer states}

In the following, for a prime number $d$ we let $\mathrm{Stab}_{d,m} \subset (\C^d)^{\otimes m}$ be the set of $m$-qudit stabilizer states, and let $s_{d,m} := \abs{\mathrm{Stab}_{d,m}}= d^m \prod_{j=1}^m (d^j+1)$ (see \cite[Corollary 21]{gross2006hudson} or \cite[Proposition 2]{kueng2015qubit}). For an integer $c\ge 1$, let $G^{(c)}_{d,m}$ be the Gram matrix of the set
\begin{equation}
\{\ket{\psi}^{\otimes c} : \ket{\psi}\in \mathrm{Stab}_{d,m}\}.
\end{equation}

\begin{proposition}\label{prop:stab_gram}
    For every prime number $d$ and all positive integers $m$ and $c$ it holds that
    \begin{equation}
        \norm{G^{(c)}_{d,m}}_{\textup{F}}^2 = s_{d,m} \prod_{j=0}^{m-1}\left(1+d^{\,j-c+2}\right).
    \end{equation}
\end{proposition}

The proof of Proposition~\ref{prop:stab_gram} is in Appendix~\ref{app:stab_gram}.

\subsubsection{\texorpdfstring{$2$}{2}-learning stabilizer states}

We now apply Proposition~\ref{prop:row_regular_2_learning} to stabilizer
states. 
Fix a prime $d$, and let
\begin{equation}
    \mathrm{Stab}_{d,m} \subset (\C^d)^{\otimes m}
\end{equation}
denote the set of $m$-qudit stabilizer states.  The Clifford group acts
transitively on $\mathrm{Stab}_{d,m}$, so for a fixed
$\ket{\psi}\in\mathrm{Stab}_{d,m}$ the quantity
\begin{equation}
\Phi_c(d,m)
:=
\sum_{\phi\in\mathrm{Stab}_{d,m}}
\abs{\braket{\psi}{\phi}}^c
\end{equation}
is independent of $\ket{\psi}$. We first determine this quantity exactly.

\begin{proposition}\label{prop:stab_absolute_moment}
For every prime $d$, every positive integer $m$, and every positive integer
$c$,
\begin{equation}
\Phi_c(d,m)
=
\prod_{j=0}^{m-1}
\left(1+d^{j+2-c/2}\right).
\end{equation}
\end{proposition}

\begin{proof}
For a fixed stabilizer state $\ket{\psi}$, let $a_r$ denote the number of
stabilizer states $\ket{\phi}$ satisfying
\begin{equation}
\abs{\braket{\psi}{\phi}}^2=d^{-r},
\qquad
r\in\{0,1,\ldots,m\}.
\end{equation}
The possible nonzero overlaps between stabilizer states have precisely this
form~\cite[Equation (41)]{obst2024wigner} (see also~\cite[Theorem 11]{garcia2017geometry}).  Hence, for every positive integer $t$,
\begin{equation}
\sum_{r=0}^m a_r d^{-rt}
=
\sum_{\phi\in\mathrm{Stab}_{d,m}}
\abs{\braket{\psi}{\phi}}^{2t}.
\end{equation}
By Proposition~\ref{prop:stab_gram} and Clifford transitivity,
\begin{equation}
\sum_{\phi\in\mathrm{Stab}_{d,m}}
\abs{\braket{\psi}{\phi}}^{2t}
=
\prod_{j=0}^{m-1}
\left(1+d^{j-t+2}\right).
\end{equation}
Therefore
\begin{equation}
\sum_{r=0}^m a_r x^r
=
\prod_{j=0}^{m-1}(1+d^{j+2}x)
\end{equation}
for the infinitely many values $x=d^{-t}$ ($t \geq 1$ is a positive integer).  Since both sides are
polynomials of degree at most $m$, this is a polynomial identity.  Setting
$x=d^{-c/2}$ gives
\begin{equation}
\Phi_c(d,m)
=
\sum_{r=0}^m a_r d^{-rc/2}
=
\prod_{j=0}^{m-1}
\left(1+d^{j+2-c/2}\right),
\end{equation}
as claimed.
\end{proof}

Combining Propositions~\ref{prop:row_regular_2_learning} and \ref{prop:stab_absolute_moment} gives an exact criterion for
$2$-learning stabilizer states.

\begin{theorem}\label{thm:stab_exact_2_learning}
Let $d$ be prime.  Then $\mathrm{Stab}_{d,m}$ is $(c,2)$-learnable if and
only if
\begin{equation}
\prod_{j=0}^{m-1}
\left(1+d^{j+2-c/2}\right)
\leq 2.
\end{equation}
\end{theorem}

This criterion allows us to determine exactly the copy complexity of 2-learning stabilizer states.

\begin{corollary}\label{cor:stab_qudit_exact}
Let $d$ be prime. The copy complexity of $2$-learning the set of
$m$-qudit stabilizer states is
\begin{equation}
c_2\left(\mathrm{Stab}_{d,m}\right)
=
\begin{cases}
4, & m=1,\\
8, & d=2,\ m=2,\\
2m+5, & d=2,\ m\geq3,\\
7, & d=3,\ m=2,\\
2m+4, & d=3,\ m\geq3,\\
2m+3, & d\geq5,\ m\geq2.
\end{cases}
\end{equation}
In particular, $c_2\left(\mathrm{Stab}_{d,m}\right)=2m+{\Theta}(1)$.
\end{corollary}

The proof of Corollary~\ref{cor:stab_qudit_exact} can be found in Appendix~\ref{app:stab_qudit_exact}.

\begin{remark}\label{rmk:stabilizer_learning}
Any $2$-learning measurement for a set $S$ can be converted into a bounded-error
identification procedure. 
First use $c_2(S)$ copies to obtain a certified pair
$\{\ket{\alpha},\ket{\beta}\}$ containing the unknown state. 
Then use $t$
additional copies and project onto $\ket{\alpha}^{\otimes t}$.  If the
unknown state is $\ket{\alpha}$, the test never errs; if it is
$\ket{\beta}$, the error probability is
\begin{equation}
\abs{\braket{\alpha}{\beta}}^{2t}.
\end{equation}
For two distinct qudit stabilizer states,
\begin{equation}
\abs{\braket{\alpha}{\beta}}^2\leq\frac1d.
\end{equation}
Thus $t=m$ additional copies reduce the error probability to at most
$d^{-m}$, giving an information-theoretic identification procedure using
$c_2(S)+m=3m+O(1)$ copies in total. This reproduces the information-theoretic $O(m)$ upper bound proven in~\cite{montanaro2017learning,aaronson2008identifying,allcock2025reconquering,CCLW26}, with the caveat that our result applies only for prime-dimensional qudits. A drawback to our result is that the 2-learning measurement may not be efficiently implementable. An advantage is that the 2-learning approach additionally yields a distinguished pair $\{\ket{\alpha},\ket{\beta}\}$ of stabilizer states guaranteed to contain the unknown state. Note that $\Omega(m)$ copies are necessary: There are $d^{\Theta(m^2)}$ stabilizer states on $m$ $d$-dimensional qudits, so identifying an unknown stabilizer state  to constant error probability requires $\Omega(m)$ copies by Holevo's theorem.
\end{remark}

\subsubsection{\texorpdfstring{$k$}{k}-learning stabilizer states}

Beyond the case $k=2$, we can use the above results to prove a general relationship between $c$ and $k$ for stabilizer states.

\begin{corollary}\label{cor:stab} 
    Let $d$ be prime and let $m$ and $c$ be positive integers. The set of stabilizer states $\mathrm{Stab}_{d,m} \subset (\C^d)^{\otimes m}$ is $(c,k)$-learnable for
    \begin{equation}\label{eq:stab_states_copy_complex}
        k = \ceil{s_{d,m}\left(1-\frac{1}{\prod_{j=0}^{m-1}(1+d^{j-c+2})}\right)
        }+1.
    \end{equation}
\end{corollary}

\begin{proof}
    By Theorem~\ref{thm:main} and Proposition~\ref{prop:stab_gram}, the set of stabilizer states is $(c,k)$-learnable as long as
    \begin{equation}
        s_{d,m} \prod_{j=0}^{m-1}\left(1+d^{j-c+2}\right) \leq \frac{s_{d,m}^2}{s_{d,m}-k+1}.
    \end{equation}
    Solving for $k$ produces the quantity in Equation~\eqref{eq:stab_states_copy_complex}. 
\end{proof}

\section{Zero-error quantum state mutation detection problems} \label{sec:applications}

We now apply Theorems~\ref{thm:main} and~\ref{thm:m_excl_threshold} to families of product states that arise in anomaly and changepoint detection. These families of pure states have a significant amount of symmetry, which allows us to obtain explicit zero-error $k$-learning criteria.

\subsection{Single anomaly detection}\label{sec:single_anomaly}

Let $\ket{\phi}$ and $\ket{\psi}$ be pure states, and suppose that exactly one of $n$ systems is prepared in the anomalous state $\ket{\phi}$, while all remaining systems are prepared in the reference state $\ket{\psi}$. If we do not know which of the $n$ systems is prepared in the anomalous state, then we know that the global state is of the form
\begin{equation}
    \ket{\psi}^{\otimes(i-1)} \otimes \ket{\phi}\otimes \ket{\psi}^{\otimes(n-i)}
\end{equation}
for some $i \in \{1,2,\ldots,n\}$. Determining the index $i$ exactly can only be done if $\braket{\phi}{\psi} = 0$. However, we can narrow $i$ down to one of at most $k$ possible values as long as $\abs{\braket{\phi}{\psi}}$ is small enough:

\begin{theorem}\label{thm:single_anomaly}
    Let $n \geq 2$ be an integer and let $k \in \{1,2,\ldots,n\}$. The set
    \begin{equation}
        \left\{\ket{\psi}^{\otimes(i-1)} \otimes \ket{\phi}\otimes \ket{\psi}^{\otimes(n-i)} : i \in \{1,2,\ldots,n\}\right\}
    \end{equation}
    is $k$-learnable if and only if
    \begin{equation}
        \abs{\braket{\phi}{\psi}} \leq \sqrt{\frac{k-1}{n-1}}.
    \end{equation}
\end{theorem}

\begin{proof}
    For $i \neq j$, the two states $\ket{\psi}^{\otimes(i-1)} \otimes \ket{\phi}\otimes \ket{\psi}^{\otimes(n-i)}$ and $\ket{\psi}^{\otimes(j-1)} \otimes \ket{\phi}\otimes \ket{\psi}^{\otimes(n-j)}$ differ only in tensor factors $i$ and $j$. Consequently, their inner product with each other is
    \begin{equation}
        \braket{\phi}{\psi} \cdot \braket{\psi}{\phi} = \abs{\braket{\phi}{\psi}}^2.
    \end{equation}
    Since this is true for all $i \neq j$, the Gram matrix of this multiset of states is
    \begin{equation}
        G = (1-\gamma)\I + \gamma \one\one^\t,
    \end{equation}
    where $\gamma = \abs{\braket{\phi}{\psi}}^2$ and $\one$ is the all-ones column vector. The result now follows immediately from \cref{cor:equi_angle}.
\end{proof}

In particular, Theorem~\ref{thm:single_anomaly} immediately implies that the smallest value of $k$ for which
\begin{equation}
    \left\{\ket{\psi}^{\otimes(i-1)} \otimes \ket{\phi}\otimes \ket{\psi}^{\otimes(n-i)} : i \in \{1,2,\ldots,n\}\right\}
\end{equation}
is $k$-learnable is $k = 1 + \ceil{(n-1)\abs{\braket{\phi}{\psi}}^2}$.

\subsection{Multi-anomaly detection}\label{sec:multi_anomaly} 

We next suppose that exactly $t \geq 1$ of the $n$ systems are anomalous. Let $\ket{\psi_0}$ and $\ket{\psi_1}$ be pure states, and consider the global states of the form
\begin{equation}
    \ket{\psi_{i_1}} \otimes \cdots \otimes \ket{\psi_{i_n}},
\end{equation}
where $i_j \in \{0,1\}$ for each $j \in \{1,2,\ldots,n\}$ and $\sum_{j=1}^n i_j = t$. We note that there are $\binom{n}{t}$ global states of this form.

\begin{theorem}\label{thm:multi-anomaly}
	Let $n \geq 2$ and $1 \leq t \leq n-1$ be integers, and define $\gamma := \abs{\braket{\psi_0}{\psi_1}}^2$. Then for each integer $1 \leq k \leq \binom{n}{t}$, the multiset of states
    \begin{equation}\label{eq:multi-anomaly_st}
        S_t := \left\{\ket{\psi_{i_1}} \otimes \ldots \otimes \ket{\psi_{i_n}} : i_1,\ldots,i_n \in \{0,1\} \ \ \text{and} \ \ \sum_{j=1}^n i_j = t\right\}
    \end{equation}
	satisfies the following:
	\begin{enumerate}
		\item[(1)] If $\sum_{r=0}^{n}\binom{t}{r}\binom{n-t}{r}\gamma^r > k$ then $S_t$ is not $k$-learnable.
		\item[(2)] If $\left(\binom{n}{t}-k+1\right)\sum_{r=0}^{n}\binom{t}{r}\binom{n-t}{r}\gamma^{2r} \leq \binom{n}{t}$ then $S_t$ is $k$-learnable.
	\end{enumerate}
\end{theorem}

Before proving this theorem, we note that if $r > t$ then we take the convention that $\binom{t}{r} = 0$. Also, in the $t = 1$ case we have
\begin{equation}
    \sum_{r=0}^{n}\binom{t}{r}\binom{n-t}{r}\gamma^r = 1 + (n-1)\gamma = 1 + (n-1)\abs{\braket{\psi_0}{\psi_1}}^2,
\end{equation}
so Theorem~\ref{thm:multi-anomaly}(1) says that if $1 + (n-1)\abs{\braket{\psi_0}{\psi_1}}^2 > k$ then $S_1$ is not $k$-learnable. This is equivalent to the $\abs{\braket{\phi}{\psi}} > \sqrt{(k-1)/(n-1)}$ condition from Theorem~\ref{thm:single_anomaly}. However, Theorem~\ref{thm:multi-anomaly}(2) in the $t = 1$ case is weaker than the $k$-learnability condition from Theorem~\ref{thm:single_anomaly}.

\begin{proof}[Proof of Theorem~\ref{thm:multi-anomaly}]
    Let $G$ be the Gram matrix of $S_t$, indexed by the $t$-element subsets of $\{1,2,\ldots,n\}$. Let $T$ and $U$ be $t$-element subsets of $\{1,2,\ldots,n\}$ and set
    \begin{equation}
        r := \abs{T \setminus U} = \abs{U \setminus T} = t - \abs{T \cap U}.
    \end{equation}
    Every tensor factor indexed by $T \setminus U$ contributes $\braket{\psi_0}{\psi_1}$ to the corresponding entry of $G$, and every tensor factor indexed by $U\setminus T$ similarly contributes $\braket{\psi_1}{\psi_0}$. It follows that
    \begin{equation}\label{eq:multi-anomaly-gram-entry}
        G_{T,U} = \left(\braket{\psi_0}{\psi_1}\braket{\psi_1}{\psi_0}\right)^r = \gamma^r.
    \end{equation}
    For each of the $\binom{n}{t}$ fixed $t$-element sets $T$, the number of $t$-element sets $U$ satisfying $\abs{T \setminus U} = r$ is $\binom{t}{r}\binom{n-t}{r}$, since we can choose $r$ elements to remove from $T$ and $r$ elements to add back in from $\{1,2,\ldots,n\} \setminus T$. Since all entries of $G$ are nonnegative, it follows that
    \begin{align}\label{eq:multi-anomaly-l1-norm}
        \sum_{i,j} |G_{i,j}| 
        & = \binom{n}{t}\sum_{r=0}^{n}\binom{t}{r}\binom{n-t}{r}\gamma^r, \quad \text{and} \\\label{eq:multi-anomaly-frobenius-norm}
        \norm{G}_{\textup{F}}^2 & = \binom{n}{t}\sum_{r=0}^{n}\binom{t}{r}\binom{n-t}{r}\gamma^{2r}.
    \end{align}
    If
    \begin{equation}
         \sum_{r=0}^{n}\binom{t}{r}\binom{n-t}{r}\gamma^r > k,
    \end{equation}
    then $\norm{\operatorname{vec}(G)}_1 > \binom{n}{t}k$, so Theorem~\ref{thm:m_excl_threshold} tells us that $S_t$ is not $k$-learnable (recall that the number of states here is $\binom{n}{t}$), proving part~(1) of this theorem. On the other hand, if
    \begin{equation}
        \left(\binom{n}{t}-k+1\right)\sum_{r=0}^{n}\binom{t}{r}\binom{n-t}{r}\gamma^{2r} \leq \binom{n}{t}
    \end{equation}
    then
    \begin{equation}
        \norm{G}_{\textup{F}}^2 \leq \frac{\binom{n}{t}^2}{\binom{n}{t}-k+1},
    \end{equation}
    so Theorem~\ref{thm:main} tells us that $S_t$ is $k$-learnable, proving part~(2) of this theorem.
\end{proof}

In particular, Theorem~\ref{thm:multi-anomaly} tells us that the smallest $k$ for which the multiset $S_t$ from Equation~\eqref{eq:multi-anomaly_st} is $k$-learnable satisfies
\begin{equation}
	\ceil{\sum_{r=0}^{n}\binom{t}{r}\binom{n-t}{r}\gamma^r} \leq k \leq \binom{n}{t} + 1 - \floor{\frac{\binom{n}{t}}{\sum_{r=0}^{n}\binom{t}{r}\binom{n-t}{r}\gamma^{2r}}}.
\end{equation}

\subsection{Changepoint detection}\label{sec:changepoint}

Once again, let $\ket{\phi}$ and $\ket{\psi}$ be pure states and let $n$ be a positive integer. Consider the states of the form $\ket{\phi}^{\otimes(i-1)} \otimes \ket{\psi}^{\otimes(n-i)}$ for some $i \in \{1,2,\ldots,n\}$ (i.e., the global states that locally equal $\ket{\phi}$ until they change to $\ket{\psi}$ at some point). There are $n$ such states, each of which is $(n-1)$-partite (if $i = 1$ then the state equals $\ket{\psi}$ on each party, and if $i = n$ then it equals $\ket{\phi}$ on each party).

\begin{theorem}\label{thm:changepoint}
	Let $n \geq 2$ be an integer and let $\gamma := \abs{\braket{\phi}{\psi}} < 1$. Then for each integer $1 \leq k \leq n$, the multiset of states
    \begin{equation}
        S := \left\{ \ket{\phi}^{\otimes(i-1)} \otimes \ket{\psi}^{\otimes(n-i)} : i \in \{1, 2, \ldots, n\} \right\}
    \end{equation}
    satisfies the following:
	\begin{enumerate}
		\item[(1)] If $2\gamma((n-1)-n\gamma+\gamma^n)/(1-\gamma)^2 > n(k-1)$ then $S$ is not $k$-learnable.
		\item[(2)] If $2\gamma^2\left((n-1)-n\gamma^2+\gamma^{2n}\right)/(1-\gamma^2)^2 \leq n(k-1)/(n-k+1)$ then $S$ is $k$-learnable.
	\end{enumerate}
\end{theorem}

\begin{proof}
	Let $G$ be the Gram matrix of $S$. If $i < j$ then the $i$-th and $j$-th members of $S$ differ precisely in the $j-i$ different tensor factors $i,i+1,\ldots,j-1$, so
	\begin{equation}
		\abs{G_{i,j}} = \abs{\braket{\phi}{\psi}}^{j-i} = \gamma^{j-i}.
	\end{equation}
	If $i > j$ then we similarly have $\abs{G_{i,j}} = \gamma^{i-j}$, so $\abs{G_{i,j}} = \gamma^{\abs{j-i}}$ for all $i$ and $j$.
    
    For each $r \in \{1, 2, \ldots, n-1\}$, there are $2(n-r)$ ordered pairs $(i,j) \in \{1, 2, \ldots, n\}^2$ with $\abs{j-i} = r$. It follows that
	\begin{align} 
        \sum_{i,j} |G_{i,j}| 
        & = n + 2\sum_{r=1}^{n-1}(n-r)\gamma^r = n + \frac{2\gamma\left((n-1)-n\gamma+\gamma^n\right)}{(1-\gamma)^2},\\
		\norm{G}_{\textup{F}}^2 & = n + 2\sum_{r=1}^{n-1}(n-r)\gamma^{2r} = n + \frac{2\gamma^2\left((n-1)-n\gamma^2+\gamma^{2n}\right)}{(1-\gamma^2)^2},
	\end{align}
    where the closed-form expressions for the sum follow from the arithmetico--geometric formula (valid whenever $x \neq 1$):
	\begin{equation}
		\sum_{r=1}^{n-1}(n-r)x^r = \frac{x\left((n-1)-nx+x^n\right)}{(1-x)^2}.
	\end{equation}

    Parts~(1) and~(2) of the theorem now follow from Theorems~\ref{thm:m_excl_threshold} and~\ref{thm:main}, respectively.
\end{proof}

\section{Typical behaviour of \texorpdfstring{$k$}{k}-learnability}\label{sec:typical} 

We now explore when one should expect multisets of randomly generated pure states to be $k$-learnable or not. Roughly speaking, there is a trade-off between the ambient dimension $d$ and the number of states $n$: the more states are packed into smaller dimensions, the harder it is to $k$-learn those states.

\begin{proposition}\label{prop:random_gram_not_klearn}
    Let $n \geq 2$, $k \in \{1, 2, \ldots, n-1\}$, $0 < \delta < 1$, and let $S = \{ \ket{\psi_1}, \ket{\psi_2}, \ldots,\ket{\psi_n} \} \subset \C^d$ be independent Haar-random pure states. If
    \begin{equation}
        d < \frac{\pi}{4}\left( \frac{k-1}{n-1} + \sqrt{\frac{\ln(1/\delta)}{n-1}} \right)^{-2}
    \end{equation}
    then these states are not $k$-learnable with probability at least $1-\delta$. In particular, there is a constant $C_{\delta}> 0$ such that if $d \leq C_{\delta} \min\{ n, (n/k)^2 \}$ then these states are not $k$-learnable with probability at least $1-\delta$.
\end{proposition} 

\begin{proof} 
    Let $G$ be the Gram matrix of the vectors $\ket{\psi_1}, \ket{\psi_2}, \ldots, \ket{\psi_n}$. 
    Then
    \begin{equation}
        \sum_{i,j} |\braket{\psi_i}{\psi_j}| 
        = n+\sum_{i\neq j}\abs{\braket{\psi_i}{\psi_j}}.
    \end{equation}
    Let
    \begin{equation}
        U := \frac{1}{n(n-1)}\sum_{i\neq j}\abs{\braket{\psi_i}{\psi_j}},
    \end{equation}
    so that $\sum_{i,j} |\braket{\psi_i}{\psi_j}| > nk$ is equivalent to $U > (k-1)/(n-1)$.
    
    Recall that Hoeffding's inequality (see Inequality~(5.7) in \cite{hoeffding1963probability}, for example) says that for every $t > 0$ we have
    \begin{align}\label{ineq:hoeffding}
        \Pr(U \leq \mu - t) \leq \exp(-2\floor{n/2} t^2),
    \end{align}
    where $\mu = \mathbb{E}[U]$ (in the terminology and notation of \cite{hoeffding1963probability}, $U$ is a ``$U$-statistic'' associated with the function $g(\ket{\psi_i},\ket{\psi_j}) = \abs{\braket{\psi_i}{\psi_j}}$).

    Fix $i\neq j$. By unitary invariance of the Haar measure, all off-diagonal entries of $G$ have the same expectation, so
    \begin{equation}
        \mu = \mathbb{E}[U] = \mathbb E\!\left[\abs{\braket{\psi_i}{\psi_j}}\right] = \frac{\sqrt{\pi}\Gamma(d)}{2\Gamma(d+1/2)} \geq \frac{\sqrt{\pi}}{2\sqrt{d}},
    \end{equation}
    where this value comes from \cite{zyczkowski2000truncations} and the
    inequality follows from Gautschi's inequality. We note that, by our
    hypothesis on $d$, we have
    \begin{equation}
        \mu \geq \frac{\sqrt{\pi}}{2\sqrt{d}} > \frac{k-1}{n-1} + \sqrt{\frac{\ln(1/\delta)}{n-1}} \geq \frac{k-1}{n-1},
    \end{equation}
    so we can choose
    \begin{equation}
        t = \mu - \frac{k-1}{n-1} > 0
    \end{equation}
    in Hoeffding's inequality.
    
    With these values of $\mu$ and $t$, Inequality~\eqref{ineq:hoeffding} tells us that
    \begin{equation}
        \begin{aligned}
            \Pr\left(U \leq \frac{k-1}{n-1}\right) & \leq \exp\left(-2\floor{\frac{n}{2}} \left(\mu - \frac{k-1}{n-1}\right)^2\right), \\
            & \leq \exp\left(-(n-1) \left(\frac{\sqrt{\pi}}{2\sqrt{d}} - \frac{k-1}{n-1}\right)^2\right) \\
            & \leq \delta,
        \end{aligned}
    \end{equation}
    with the final inequality following from our hypothesis on $d$. It follows that
    \begin{equation}
        \Pr\left(\sum_{i,j} |\braket{\psi_i}{\psi_j}| > nk\right) = \Pr\left(U > \frac{k-1}{n-1}\right) \geq 1-\delta
    \end{equation}
    which (thanks to Theorem~\ref{thm:m_excl_threshold}) completes the proof.
\end{proof} 

The previous proposition showed that if $d$ is small (compared to $n$) then we do not expect randomly generated states to be $k$-learnable. The following proposition shows that, by contrast, if $d$ is large then we \emph{do} expect randomly generated states to be $k$-learnable.

\begin{proposition}\label{prop:random_gram_markov}
    Let $n \geq 2$, $k \in \{2,\ldots,n\}$, $0 < \delta < 1$, and let $S = \{ \ket{\psi_1}, \ket{\psi_2}, \ldots, \ket{\psi_n} \} \subset \C^d$ be independent Haar-random pure states. 
    If
    \begin{equation}
        d \geq \frac{(n-1)(n-k+1)}{\delta(k-1)}
    \end{equation}
    then these states are $k$-learnable with probability at least $1-\delta$. In particular, there is a constant $C_{\delta}>0$ such that if $d\geq C_{\delta} (n(n - k + 1) / k)$ then these states are $k$-learnable with probability at least $1-\delta$.
\end{proposition}

\begin{proof} 
    Let $G$ be the Gram matrix of the vectors $\ket{\psi_1}, \ket{\psi_2}, \ldots, \ket{\psi_n}$. Then
    \begin{equation}
        \norm{G}_{\textup{F}}^2 = n + \sum_{i\neq j}\abs{\braket{\psi_i}{\psi_j}}^2.
    \end{equation}
    Let
    \begin{equation}
        Y := \sum_{i\neq j}\abs{\braket{\psi_i}{\psi_j}}^2.
    \end{equation}
    By Theorem~\ref{thm:main} it suffices to prove that $Y \leq n(k-1)/(n-k+1)$ with probability at least $1-\delta$.
    
    Fix $i \neq j$. Since $\ket{\psi_i}$ and $\ket{\psi_j}$ are independent and Haar-random, we can write
    \begin{equation}
        \mathbb E\!\left[\abs{\braket{\psi_i}{\psi_j}}^2\right] = \int \int \abs{\braket{\x}{\y}}^2 d\mu(\y)d\mu(\x),
    \end{equation}
    where $\mu$ denotes Haar measure on the unit sphere in $\C^d$. Fix a unit vector $\x \in \C^d$. There exists a unitary matrix $U$ such that $U\x = \e_1$ (the first standard basis vector). Using unitary invariance of the Haar measure,
    \begin{equation}
        \int \abs{\braket{\x}{\y}}^2 d\mu(\y) = \int \abs{\braket{U\x}{U\y}}^2 d\mu(\y) = \int \abs{\braket{\e_1}{\y}}^2 d\mu(\y).
    \end{equation}
    We thus see that the inner integral is independent of $\x$, so
    \begin{equation}
        \mathbb E\!\left[\abs{\braket{\psi_i}{\psi_j}}^2\right] = \int \abs{\braket{\e_1}{\y}}^2 d\mu(\y).
    \end{equation}
    Write $\y = (\alpha_1, \alpha_2, \ldots, \alpha_d)$. Then $\sum_{\ell=1}^d \abs{\alpha_\ell}^2 = 1$, so (by unitary invariance of the Haar measure) all coordinates have the same expectation, which gives
    \begin{equation}
        \mathbb E\!\left[\abs{\braket{\psi_i}{\psi_j}}^2\right]=\mathbb E\!\left[\abs{\alpha_1}^2\right]=\frac{1}{d}.
    \end{equation}
    It follows that $\mathbb E[Y] = n(n-1)/d$. By Markov's inequality,
    \begin{equation}
        \Pr\!\left(Y > \frac{n(k-1)}{n-k+1}\right) \leq \frac{\mathbb E[Y]}{n(k-1)/(n-k+1)} = \frac{(n-1)(n-k+1)}{d(k-1)}.
    \end{equation}
    The stated bound on $d$ ensures that this quantity is at most $\delta$, yielding the result.
\end{proof}

\section*{Acknowledgements}

J.S. thanks Richard Cleve and Ankith Mohan for helpful discussions. The authors used Fable-5 and GPT-5.6 Sol to prove Lemmas~\ref{lem:boundary} and~\ref{lem:main2}, which are central ingredients for the main result. Large language models were also used for proof ideas and editorial polishing throughout the paper. The authors take full responsibility for the correctness, exposition, and attribution in the final manuscript. 
N.J.\ acknowledges support from NSERC Discovery Grant number RGPIN-2022-04098. 
B.L.\ acknowledges support from NSERC Discovery Grant number RGPIN-2026-05413. 
J.S.\ acknowledges support from the National Science Foundation Award Number 2542721. 

\bibliographystyle{alpha}
\bibliography{references}

\appendix 

\section{Deferred proofs} 

\subsection{Proof of Proposition~\ref{prop:row_regular_2_learning}}  
\label{app:prop4} 

\begin{proof}
Let $G$ be the Gram matrix of $S$. Suppose first that $R_i\leq 2$ for all $i$. Since $G_{i,i}=1$ for every $i$,
\begin{equation}
\sum_{j\neq i}\abs{G_{i,j}}
=
R_i-1
\leq 1.
\end{equation}
Thus $G$ is diagonally dominant.  We can write
\begin{equation}
G
=
\sum_{i=1}^n
\left(
1-\sum_{j\neq i}\abs{G_{i,j}}
\right)\e_i \e_i^*
+
\sum_{1\leq i<j\leq n} B^{i,j},
\end{equation}
where $B^{i,j}$ is supported only on rows and columns $i,j$ and has principal
$2\times 2$ block
\begin{equation}
\begin{pmatrix}
\abs{G_{i,j}} & G_{i,j}\\
\overline{G_{i,j}} & \abs{G_{i,j}}
\end{pmatrix}.
\end{equation}
Note that each $B^{i,j}$ is positive semidefinite. Hence $G$ is a sum of positive
semidefinite matrices supported on at most two coordinates, so it is
$2$-incoherent.  By Proposition~\ref{prop:klearnable}, $S$ is
$2$-learnable.

Now suppose that $S$ is $2$-learnable.  By
Theorem~\ref{thm:m_excl_threshold},
\begin{equation}
\sum_{i,j=1}^n \abs{G_{i,j}}
\leq 2n,
\end{equation}
or equivalently, $\sum_{i=1}^n R_i \leq 2n$. 
\end{proof}

\subsection{Proof of Theorem~\ref{thm:main}} \label{app:main_proof}

Define $r := n-k+1$. For $\d \in \R^n$, let $e_j(\d)$ denote the $j$-th elementary symmetric polynomial of the coordinates of $\d$:
\begin{equation}
    e_j(\d) := \sum_{1 \leq i_1 < \cdots < i_j \leq n} \left(\prod_{\ell=1}^j d_{i_{\ell}}\right),
\end{equation}
with the convention that $e_0(\d) = 1$. Define the cone
\begin{equation}
    \Gamma_k^{(n)} := \left\{ \d \in \R^n : e_1(\d) \geq 0,\ldots, e_k(\d) \geq 0 \right\},
\end{equation}
and the compact convex set
\begin{equation}
    \Delta_r^{(n)} := \left\{ \p \in \R^n_{\ge 0} : \sum_{j=1}^n p_j = 1, \norm{\p}_2 \le \frac{1}{\sqrt{r}} \right\}.
\end{equation}
It is known \cite{garding1959inequality} that $\Gamma_k^{(n)}$ is a closed convex cone and that $\R^n_{\geq 0} \subseteq \Gamma_k^{(n)}$. We use $\left(\Gamma_k^{(n)}\right)^\circ$ to denote the dual cone of $\Gamma_k^{(n)}$:
\begin{equation}
    \left(\Gamma_k^{(n)}\right)^\circ := \left\{ \y \in \R^n : \x \cdot \y \geq 0 \ \text{for all} \ \x \in \Gamma_k^{(n)} \right\}.
\end{equation}

Before proving the result, we need a pair of helper lemmas.

\begin{lemma}\label{lem:boundary}
    Let $n \geq 4$ be an integer and let $k \in \{3,4,\ldots,n-1\}$. Define $r := n-k+1$ and suppose $\p \in \R^n$ satisfies $\norm{\p}_2^2 = 1/r$, $\sum_{i=1}^n p_i = 1$, and $p_i > 0$ for all $i \in \{1,2,\ldots,n\}$. Then $\one-r\p \notin \Gamma_k^{(n)}$.
\end{lemma}

\begin{proof}
    Define $\x := \one-r\p$ and suppose (for the purpose of establishing a contradiction) that $\x \in \Gamma_k^{(n)}$. Notice that $e_1(\x) = \sum_{i=1}^n x_i = n-r = k-1$ and
    \begin{equation}
        \sum_{i=1}^n x_i^2 = \sum_{i=1}^n (1-r p_i)^2 = \sum_{i=1}^n (1-2rp_i+r^2 p_i^2) = n-2r+r = n-r = k-1.
    \end{equation}
    Our first goal is to prove that
    \begin{equation}
        k e_k(\x) = -r\sum_{i=1}^n p_i x_i^2e_{k-3}(\x_{\widehat{i}}),
    \end{equation}
    where $\x_{\widehat{i}} \in \R^{n-1}$ is the vector $\x$ with the $i$-th entry removed. To this end, recall the standard identities
    \begin{equation}
        ke_k(\x) = \sum_{i=1}^n x_i e_{k-1}(\x_{\widehat{i}}) \quad \text{and} \quad e_{k-1}(\x_{\widehat{i}}) = e_{k-1}(\x) - x_i e_{k-2}(\x_{\widehat{i}})
    \end{equation}
    for all $i \in \{1, 2, \ldots, n\}$. From here, we have
    \begin{align}
        ke_k(\x) & = \sum_{i=1}^n x_i e_{k-1}(\x_{\widehat{i}})\\
        & = \sum_{i=1}^n x_i\left(e_{k-1}(\x) - x_i e_{k-2}(\x_{\widehat{i}})\right)\\
        & = (k-1) e_{k-1}(\x) - \sum_{i=1}^n x_i^2 e_{k-2}(\x_{\widehat{i}})\\
        & = -\sum_{i=1}^n (x_i^2-x_i) e_{k-2}(\x_{\widehat{i}}).
    \end{align}
    Note that $x_i^2 - x_i = -rp_i x_i$, so the above string of equations becomes
    \begin{equation}
        ke_k(\x) = r\sum_{i=1}^n p_i x_i e_{k-2}(\x_{\widehat{i}}).
    \end{equation}
    Using the recursion $e_{k-2}(\x_{\widehat{i}}) = e_{k-2}(\x) - x_i e_{k-3}(\x_{\widehat{i}})$, we obtain
    \begin{equation}\label{eq:recurse_kmin3}
        ke_k(\x) = r\sum_{i=1}^n p_ix_i\left(e_{k-2}(\x) - x_i e_{k-3}(\x_{\widehat{i}})\right).
    \end{equation}
    Now observe that
    \begin{equation}
        \sum_{i=1}^n p_i x_i = \sum_{i=1}^n p_i(1 - rp_i) = \sum_{i=1}^n p_i - r\sum_{i=1}^n p_i^2 = 1 - \frac{r}{r} = 0.
    \end{equation}
    Substituting into Equation~\eqref{eq:recurse_kmin3} (finally) gives
    \begin{equation}\label{eq:kek_final}
        ke_k(\x) = -r\sum_{i=1}^n p_ix_i^2 e_{k-3}(\x_{\widehat{i}}),
    \end{equation}
    as desired.

    By \cite[Lemma~2]{johnston2022absolutely}, deleting one coordinate from a vector in $\Gamma_k^{(n)}$ produces a vector in $\Gamma_{k-1}^{(n-1)}$. It follows that every term on the right-hand side of~\eqref{eq:kek_final} (ignoring the leading negative sign) is nonnegative. It follows that $e_k(\x) \leq 0$. Since $\x \in \Gamma_k^{(n)}$, we must actually have $e_k(\x) = 0$, from which Equation~\eqref{eq:kek_final} implies
    \begin{equation}\label{eq:ek_vanish}
        e_{k-3}(\x_{\widehat{i}})=0 \quad \text{whenever} \quad x_i \neq 0.
    \end{equation}
    
    Now let $\y \in \R^m$ be the vector of nonzero entries of $\x$. Every zero entry of $\x$ must be in the same location as an entry of $\p$ that is equal to $1/r$. Since $\p$ has full support, $\x$ has at most $r-1$ zero entries, so $m \geq n-r+1 = k$. Equation~\eqref{eq:ek_vanish} therefore says that $e_{k-3}(\y_{\widehat{i}}) = 0$ for all $i \in \{1, 2, \ldots, m\}$.
    
    This is impossible. Indeed, it is not possible that $y_i \neq 0$ for all $i \in \{1, 2, \ldots, m\}$ and there exists $0 \leq \ell \leq m-1$ for which $e_\ell(\y_{\widehat{i}}) = 0$ for all $i \in \{1, 2, \ldots, m\}$. The $\ell = 0$ case is not possible simply because $e_0(\y_{\widehat{i}}) = 1$. For $\ell \geq 1$, recall again the identity
    \begin{equation}
        e_\ell(\y) = e_{\ell}(\y_{\widehat{i}}) + y_i e_{\ell-1}(\y_{\widehat{i}})
    \end{equation}
    for all $i \in \{1, 2, \ldots, m\}$. Since $e_\ell(\y_{\widehat{i}}) = 0$ by assumption, this becomes $e_\ell(\y) = y_i e_{\ell-1}(\y_{\widehat{i}})$. Summing over all $i$ gives $m e_\ell(\y) = \ell e_\ell(\y)$, so $e_\ell(\y) = 0$ and then $e_{\ell-1}(\y_{\widehat{i}})=0$ for every $i$. Descending in $\ell$ eventually gives $e_0(\y_{\widehat{i}}) = 0$, contradicting $e_0(\y_{\widehat{i}}) = 1$ and proving the lemma.
\end{proof}

\begin{lemma}\label{lem:main2}
    Suppose $n \geq 2$ is an integer and $k \in \{1, 2, \ldots, n-1\}$, and let $r := n-k+1$. Then $\Delta_r^{(n)} \subseteq (\Gamma_k^{(n)})^\circ$.
\end{lemma}

\begin{proof}
    We need to show that if $\d \in \Gamma_k^{(n)}$ and $\p \in \Delta_r^{(n)}$ then $\d \cdot \p \geq 0$. We use induction on $k$. For $k = 1$, the Cauchy--Schwarz inequality forces every $\p \in \Delta_n^{(n)}$ to equal $\one/n$, so $\d \cdot \p = e_1(\d)/n \geq 0$. For $k = 2$, the result is exactly the description of $(\Gamma_2^{(n)})^\circ$ from \cite[Theorem~6]{johnston2022absolutely}.
    
    Now let $k \geq 3$, fix $\d \in \Gamma_k^{(n)}$, and let $\p^\star$ minimize $\d \cdot \p$ over $\Delta_r^{(n)}$. Put $I := \supp(\p^\star)$ and $m := \abs{I}$. The Cauchy--Schwarz inequality gives $m\geq r$. We now split into two cases depending on the size of $m$:

    \noindent\underline{\textbf{Case 1:} $m < n$.} Let $\d_I$ and $\p_I^\star$ denote the restrictions of $\d$ and $\p^\star$, respectively, to the entries indexed by $I$. Define $k^\prime := m-r+1 = k-(n-m)$. Iterating \cite[Lemma~2]{johnston2022absolutely} gives $\d_I \in \Gamma_{k^\prime}^{(m)}$, while $\p_I^\star \in \Delta_r^{(m)}$ and $r = m - k^\prime + 1 \geq 2$, so $m > k^\prime$. Since $1 \leq k^\prime < k$, the inductive hypothesis yields
    \begin{equation}
        \d \cdot \p^\star = \d_I \cdot \p_I^\star \geq 0.
    \end{equation}
    
    \noindent\underline{\textbf{Case 2:} $m = n$.} Then $p_i^\star > 0$ for all $i$, so the nonnegativity constraints are inactive. Since $\one/n$ is strictly feasible, the Karush--Kuhn--Tucker conditions give numbers $\tau \in \R$ and $s \geq 0$ such that
    \begin{equation}
        \d = \tau\one - s\p^\star.
    \end{equation}
    If $s = 0$ then $\d = \tau\one$, so $n\tau = e_1(\d) \geq 0$ implies $\d \cdot \p^\star = \tau \geq 0$.
    
    On the other hand, if $s > 0$ then $n\tau - s = e_1(\d) \geq 0$ implies $\tau > 0$. Assume (for the sake of establishing a contradiction) that $\d \cdot \p^\star < 0$. Then the fact that $\norm{\p^\star}_2^2 = 1/r$ implies
    \begin{equation}
        0 > \d \cdot \p^\star = \tau - \frac{s}{r},
    \end{equation}
    so $\beta := s/\tau > r$. Since $\one - \beta \p^\star = \d/\tau \in \Gamma_k^{(n)}$ and $\one \in \Gamma_k^{(n)}$, convexity of $\Gamma_k^{(n)}$ gives
    \begin{equation}
        \one - r\p^\star = \left(1-\frac r\beta\right)\one + \frac r\beta\left(\one-\beta\p^\star\right) \in \Gamma_k^{(n)},
    \end{equation}
    contradicting Lemma~\ref{lem:boundary}. 
    It follows that $\d \cdot \p^\star \geq 0$, which completes the induction and the proof.
\end{proof}

\begin{proof}[Proof of Theorem~\ref{thm:main}]
    The theorem is trivial when $n \leq 2$ or $k \in \{1,n\}$, so we assume throughout the proof that $n \geq 3$ and $k \in \{2,3,\ldots,n-1\}$. Let $G$ be a Gram matrix with $\norm{G}_{\textup{F}} \leq n/\sqrt{r}$ and define the matrix $X := G/n$. Then $X$ is a density matrix (positive semidefinite with trace $1$) and
    \begin{equation}
        \norm{X}_{\textup{F}}^2 = \frac{\norm{G}_{\textup{F}}^2}{n^2} \leq \frac{1}{r}.
    \end{equation}
    It follows that $\p := \lambda(X) \in \Delta_r^{(n)}$. By Lemma~\ref{lem:main2}, this implies $\p \in \left(\Gamma_k^{(n)}\right)^\circ$. It follows from \cite[Theorem~3]{johnston2022absolutely} that $X$  is $k$-incoherent, so the same is true of $G$. By Proposition~\ref{prop:klearnable}, we conclude that every multiset of states with Gram matrix $G$ is $k$-learnable, which completes the proof.
\end{proof}

\subsection{Proof of Theorem~\ref{thm:m_excl_threshold}} 

\begin{proof} We prove the contrapositive: if $S$ is $k$-learnable then
\begin{equation}\label{eq:thm2_contra}
    \sum_{i,j = 1}^{n} \abs{\braket{\psi_i}{\psi_j}} \leq nk.
\end{equation}
To this end, let $G$ be the Gram matrix of $S$ and define
\begin{equation}
    Y = k\I - E,
\end{equation}
where, for all $1 \leq i,j \leq n$ (even if $i = j$), the $(i,j)$-entry of $E$ is the complex number with modulus $1$ and phase equal to that of $\braket{\psi_i}{\psi_j}$ (if $\braket{\psi_i}{\psi_j} = 0$ then set $E_{i,j} = 1$). Each $k \times k$ principal submatrix of $Y$ is diagonally dominant and thus positive semidefinite. It follows that, in the terminology of \cite{blekherman2022hyperbolic,johnston2022absolutely}, $Y$ is ``$k$-locally positive semidefinite'' and thus in the dual cone of the set of $k$-incoherent matrices. Since $G$ is $k$-incoherent by \cref{prop:klearnable}, this tells us that $\tr(YG) \geq 0$. Direct calculation then shows that
\begin{align}
    0 \leq \tr(YG) & = \tr\left( \left(k\I - E\right)G \right) = n(k-1) - \sum_{\substack{i,j=1 \\ i\neq j}}^{n}\abs{\braket{\psi_i}{\psi_j}},
\end{align}
which is equivalent to Inequality~\eqref{eq:thm2_contra}. This completes the proof.
\end{proof}

\subsection{Proof of Proposition~\ref{prop:stab_gram}}\label{app:stab_gram}

\begin{proof}
Fix $c \geq 1$ and define
\begin{equation}
    h_{d,m} := \norm{G^{(c)}_{d,m}}_{\textup{F}}^2 = \sum_{\psi,\phi\in\mathrm{Stab}_{d,m}}\abs{\braket{\psi}{\phi}}^{2c}.
\end{equation}
Let $\mathcal{F}_c(d,m)$ be the $c$-th \emph{frame potential} of the uniform distribution over $\mathrm{Stab}_{d,m}$, defined by
\begin{equation}\label{eq:Hm_framepotential}
    h_{d,m} = s_{d,m}^2\mathcal{F}_c(d,m).
\end{equation}
By~\cite[Theorem 2]{kueng2015qubit}, we have
\begin{equation}\label{eq:Fc_rec}
    \frac{\mathcal{F}_c(d,m+1)}{\mathcal{F}_c(d,m)} = \frac{d^{m-(c-2)}+1}{d(d^{m+1}+1)}.
\end{equation}
Note that
\begin{equation}\label{eq:Nm_rec}
    \frac{s_{d,m+1}}{s_{d,m}} = d(d^{m+1}+1).
\end{equation}
Using \eqref{eq:Hm_framepotential}, \eqref{eq:Fc_rec} and \eqref{eq:Nm_rec} we compute
\begin{equation}
    \begin{aligned}
    \frac{h_{d,m+1}}{h_{d,m}}
    &= \left(\frac{s_{d,m+1}}{s_{d,m}}\right)^2
    \frac{\mathcal{F}_c(d,m+1)}{\mathcal{F}_c(d,m)}\\
    &=\left(d(d^{m+1}+1)\right)^2
    \frac{d^{m-(c-2)}+1}{d(d^{m+1}+1)}
    =
    d(d^{m+1}+1)\left(d^{m-(c-2)}+1\right).
    \end{aligned}
\end{equation}
Define $t_{d,m} := h_{d,m} / s_{d,m}$ so that
\begin{equation}\label{eq:Km_rec}
    t_{d,m+1}
    =
    \frac{h_{d,m+1}}{s_{d,m+1}}
    =
    \frac{h_{d,m} d(d^{m+1}+1)\left(d^{m-(c-2)}+1\right)}
    {s_{d,m} d(d^{m+1}+1)}
    =
    t_{d,m}\left(d^{m-(c-2)}+1\right).
\end{equation}
    
For the base case $m=1$, \cite[Theorem 2]{kueng2015qubit} gives
\begin{equation}
    \mathcal{F}_c(d,1)=\frac{d^{2-c}+1}{d(d+1)}.
\end{equation}
Since $s_{d,1}=d(d+1)$, we obtain
\begin{equation}
    h_{d,1}
    =
    s_{d,1}^2 \mathcal{F}_c(d,1)
    =
    (d(d+1))^2 \frac{d^{2-c}+1}{d(d+1)}
    =
    d(d+1)\left(d^{2-c}+1\right),
\end{equation}
hence
\begin{equation}
    t_{d,1}=\frac{h_{d,1}}{s_{d,1}}=d^{2-c}+1.
\end{equation}

For $m \geq 2$, iterating \eqref{eq:Km_rec} yields
\begin{equation}
    t_{d,m} = \left(d^{2-c}+1\right)\prod_{j=1}^{m-1}\left(1+d^{j-(c-2)}\right) = \prod_{j=0}^{m-1}\left(1+d^{j-(c-2)}\right),
\end{equation}
which gives
\begin{equation}
    h_{d,m} = s_{d,m} t_{d,m} = s_{d,m} \prod_{j=0}^{m-1}\left(1+d^{j-(c-2)}\right),
\end{equation}
as claimed.
\end{proof}

\subsection{Proof of Corollary~\ref{cor:stab_qudit_exact}}\label{app:stab_qudit_exact}

\begin{proof}
By Theorem~\ref{thm:stab_exact_2_learning}, $\mathrm{Stab}_{d,m}$ is $(c,2)$-learnable if and only if
\begin{equation}
    \prod_{j=0}^{m-1}\left(1+d^{j+2-c/2}\right) \leq 2.
\end{equation}

First suppose $m = 1$. For $c=4$ the left-hand side is $1+d^{2-4/2}=2$, whereas for $c=3$ it is $1+d^{1/2} > 2$, so $c_2(\mathrm{Stab}_{d,1}) = 4$ for every prime $d$.

Now suppose $m \geq 2$. For $c=2m+2$,
\begin{equation}
    \prod_{j=0}^{m-1}\left(1+d^{j+2-c/2}\right) = \prod_{r=1}^{m}\left(1+d^{1-r}\right).
\end{equation}
The factor corresponding to $r=1$ equals $2$, while all remaining factors are strictly greater than $1$. It follows that
\begin{equation}
    \prod_{r=1}^{m}\left(1+d^{1-r}\right) > 2,
\end{equation}
and therefore $c_2(\mathrm{Stab}_{d,m}) \geq 2m+3$ for all $m \geq 2$.

We now split into cases to consider the different possible local dimensions $d$.

\noindent\underline{\textbf{Case 1:} $d\geq5$.}
Set $c=2m+3$. Then
\begin{equation}
\prod_{j=0}^{m-1}
\left(1+d^{j+2-c/2}\right)
=
\prod_{r=1}^{m}
\left(1+d^{-r+1/2}\right).
\end{equation}
Using the inequality $\ln(1+x) \leq x$, we see that
\begin{equation}
    \ln\left(\prod_{r=1}^{m}\left(1+d^{-r+1/2}\right)\right) \leq \sum_{r=1}^{\infty}d^{-r+1/2} = \frac{\sqrt d}{d-1} \leq \frac{\sqrt5}{4} < \ln(2).
\end{equation}
It follows that
\begin{equation}
\prod_{r=1}^{m}
\left(1+d^{-r+1/2}\right)<2,
\end{equation}
so $2m+3$ copies suffice. Since we already showed that $c_2(\mathrm{Stab}_{d,m}) \geq 2m+3$, we conclude that $c_2(\mathrm{Stab}_{d,m}) = 2m+3$ whenever $d \geq 5$ and $m \geq 2$.

\noindent\underline{\textbf{Case 2:} $d=3$.}
For $c=2m+3$, we have
\begin{equation}
    \prod_{j=0}^{m-1}\left(1+3^{j+2-c/2}\right) = \prod_{r=1}^{m}\left(1+3^{-r+1/2}\right).
\end{equation}
When $m=2$, this product becomes
\begin{equation}
    \left(1+\frac{1}{\sqrt3}\right)\left(1+\frac{1}{3\sqrt3}\right) < 2,
\end{equation}
so $7$ copies suffice. Since we already showed that $c_2(\mathrm{Stab}_{3,2}) \geq 2m+3 = 7$, we conclude that $c_2(\mathrm{Stab}_{3,2}) = 7$.

On the other hand, when $m \geq 3$, this product can be bounded by
\begin{equation}
    \prod_{r=1}^{m}\left(1+3^{-r+1/2}\right) \geq \left(1+\frac{1}{\sqrt3}\right)\left(1+\frac{1}{3\sqrt3}\right)\left(1+\frac{1}{9\sqrt3}\right) > 2,
\end{equation}
so $2m+3$ copies do not suffice. On the other hand, for $c=2m+4$,
\begin{equation}
    \prod_{j=0}^{m-1}\left(1+3^{j+2-c/2}\right) = \prod_{r=1}^{m}(1+3^{-r}),
\end{equation}
and
\begin{equation}
    \ln\left(\prod_{r=1}^{m}(1+3^{-r})\right) \leq \sum_{r=1}^{\infty}3^{-r} = \frac{1}{2} < \ln(2),
\end{equation}
so $2m+4$ copies suffice. It follows that $c_2(\mathrm{Stab}_{3,m}) = 2m+4$ whenever $m \geq 3$.

\noindent\underline{\textbf{Case 3:} $d=2$.}
For $m=2$, direct substitution into Theorem~\ref{thm:stab_exact_2_learning} gives $c_2(\mathrm{Stab}_{2,2}) = 8$. Now suppose $m \geq 3$. If $c=2m+4$, then
\begin{equation}
    \prod_{j=0}^{m-1}\left(1+2^{j+2-c/2}\right) = \prod_{r=1}^{m}(1+2^{-r}) \geq \left(1+\frac12\right)\left(1+\frac14\right)\left(1+\frac18\right) > 2.
\end{equation}
It follows that $2m+4$ copies do not suffice.

For $c = 2m+5$,
\begin{equation}
    \prod_{j=0}^{m-1}\left(1+2^{j+2-c/2}\right) = \prod_{r=1}^{m} \left(1+2^{-r-1/2}\right) < \prod_{r=1}^{\infty}\left(1+2^{-r-1/2}\right).
\end{equation}
Again using $\ln(1+x)\leq x$, we obtain
\begin{equation}
    \prod_{r=1}^{\infty}\left(1+2^{-r-1/2}\right) = \left(1+2^{-3/2}\right)\prod_{r=2}^{\infty}\left(1+2^{-r-1/2}\right) < 2,
\end{equation}
so $2m+5$ copies suffice. It follows that $c_2(\mathrm{Stab}_{2,m}) = 2m+5$ whenever $m \geq 3$, which completes the proof. 
\end{proof}
\end{document}